\documentclass[runningheads]{llncs}
\usepackage[T1]{fontenc}
\usepackage{graphicx}
\usepackage[table,xcdraw]{xcolor}
\usepackage{amsmath,mathtools}    
\usepackage{amssymb}
\usepackage{tikz}
\usepackage[ruled,linesnumbered,noend]{algorithm2e}
\usepackage{enumerate}
\usepackage{hyperref}
\usepackage{url}
\usepackage[capitalize]{cleveref}
\usepackage{siunitx}
\usepackage[]{booktabs}

\usetikzlibrary{arrows}
\usetikzlibrary{shapes}
\usetikzlibrary{automata}
\usetikzlibrary{positioning,shadows,trees}

\usepackage{multirow}
\usepackage{thm-restate}
\usepackage[colorinlistoftodos]{todonotes}
\usepackage{listings}
\usepackage{subcaption}
\usepackage{lineno}
\usepackage{paralist}
\usepackage{makecell}

\usepackage[framemethod=TikZ]{mdframed}
\newcommand{\jal}[1]{\textcolor{pink}{\ifmmode \text{[JA: #1]}\else [JA: #1] \fi}}
\newcommand{\dd}[1]{\textcolor{cyan}{\ifmmode \text{[DD: #1]}\else [DD: #1] \fi}}
\newcommand{\yc}[1]{\textcolor{teal}{\ifmmode \text{[YC: #1]}\else [YC: #1] \fi}}

\newcommand{\hide}[1]{}

\newcommand{\complexityFont}[1]{\textsc{#1}}
\newcommand{\NP}[0]{\complexityFont{NP}\xspace}
\newcommand{\PSPACE}[0]{\complexityFont{PSpace}\xspace}

\newcommand{\positive}[0]{\mathbb{N}_{>0}\xspace}

\newcommand{\qfac}{QFAC\xspace}
\newcommand{\bqfac}{BQFAC\xspace}

\newcommand{\qfa}{QFA\xspace}

\newcommand{\fixlanclass}[2]{\mathcal{L}_{#1 #2}}
\newcommand{\lanclass}[1]{\mathcal{L}_{#1}}

\newcommand{\cmp}{\mathsf{Cmp}\xspace}

\newcommand{\val}{\textit{Pr}}

\newcommand{\abs}[1]{\left|#1\right|}
\newcommand{\norm}[1]{\abs{\abs{#1}}_2}

\newcommand{\floor}[1]{\left\lfloor#1\right\rfloor}

\DeclareMathOperator*{\argmax}{argmax}
\DeclareMathOperator*{\opt}{opt}

\newcommand{\ket}[1]{|#1\rangle}

\newcommand{\bra}[1]{\langle#1|}

\renewcommand{\Bar}[1]{\,\overline{#1}\,}
\renewcommand{\Tilde}[1]{\,\widetilde{#1}\,}

\newcommand{\mybraket}[1]{\langle#1\rangle}

\newcommand{\lqfa}{\lanclass{\textit{\qfa}}}
\newcommand{\lqfac}{\lanclass{\textit{\qfac}}}
\newcommand{\lbqfac}{\lanclass{\textit{\bqfac}}}

\begin{document}
\title{The Emptiness Problem for Quantum Finite Automata with Classical States}
 
%
%
\author{Jyun-Ao Lin\inst{1}\orcidID{0000-0001-8560-2147} \and
Patrick Totzke\inst{2}\orcidID{0000-0001-5274-8190} \and
Yun Chen Tsai\inst{3}\orcidID{0009-0003-7705-9609} \and
Di-De Yen\inst{2}\orcidID{0000-0003-0045-9594}}

\authorrunning{J.A. Lin, P. Totzke, Y.C. Tsai, and D.D. Yen}

\institute{National Taipei University of Technology, Taiwan \and
School of Computer Science and Informatics, University of Liverpool, UK \and
National Institute of Informatics and SOKENDAI, Japan}

%
%
%
\maketitle              
%

\begin{abstract}
Quantum Finite Automata with Classical states (QFACs) are nondeterministic finite automata over a finite alphabet of quantum operations. We study expressiveness of this model on finite words and the corresponding emptiness problem. We show that regular languages are incomparable with those definable by Quantum Finite Automata (QFAs) and that both are strictly subsumed by QFAC-definable languages. 

We show that the emptiness problem for a QFAC can be reduced to the emptiness of the language intersection of a QFA and a finite automaton. This intersection is known to be decidable for strict thresholds but undecidable for non-strict cases. Furthermore, we consider the problem for flat QFACs, a restriction where the underlying automata contain no nested loops, and relate it to the higher-dimensional orbit problem, a long-standing open challenge in dynamical systems.

Finally, we propose a sound and semi-complete witness searching procedure to verify the non-emptiness of one-loop QFACs, which are sufficiently expressive to represent some prominent quantum algorithms, such as Grover's search and quantum random walks.

\end{abstract}

\section{Introduction}
Formal verification of critical hardware and software is well established, both at circuit \cite{Burch:94:IEEE,Clarke:86,Drechsler:04,Wolfgang:03,Kimura:20,Seligman:23} and program level \cite{Vardi:86,Sieber:13,Fetzer:88,Srivastava:10,Vardi:21}. For quantum systems, circuit verification \cite{Tsai:21,Amy:19,Burgholzer:21,Chen:26,Chen:23,Chen:26:TACAS,Hu:25,Maidl:25} has received far more attention than program-level verification \cite{Chareton:21,Ying:23,Lewis:24:ACM,Chen:25:TACAS}, where current methods 
require user-supplied proofs or annotations.

In the automated verification of classical programs, state machines or automata are frequently used to simulate program behavior \cite{Baier:08:text}. 
The added complication in quantum programs is that they combine classical control flow with operations on quantum states. Therefore, in order to replicate the automata-based approach to verification, we need a model that captures both aspects. 
Several such automata models have been proposed for quantum computation, including quantum finite automata (QFAs) \cite{Moore:00:TCS,Blondel:05:SIAM,Bertoni:03:DLT}, quantum finite automata with control languages (QFALs) \cite{Mereghetti:06:RAIRO}, and quantum finite automata with classical states (QFACs) \cite{Qiu:15:JCSS,Zheng:12:LA}. The latter two generalize the basic QFA model which, perhaps surprisingly, does not incorporate a finite-state control and instead corresponds to a single state with one or more loops that determine operations on the quantum state. 
In this work, we focus on QFACs, in which the classical states simulate and track program control flow while quantum operations act on the quantum state throughout the computation. A QFAC acts as a language acceptor over an alphabet of quantum operations; a word is accepted if there exists a run such that the final quantum state, derived by applying the operations dictated by the word starting in a fixed initial quantum state, lies within a target space.
This allows to study decision problems such as equivalence, minimization, and emptiness to argue about the represented quantum programs. 
Among these, the emptiness problem is the most fundamental and corresponds to the problem to determine whether a specific program configuration is reachable.

\begin{example}
Consider the quantum program in \cref{Fig:c_code}, represented by the QFAC in \cref{Fig:c_code_automaton}, which preserves the program flow as well as semantics directly: The single qubit state \lstinline{z} is initialized as $\ket{0}$. In the first for-loop, at each iteration, the quantum state is rotated by a given angle $\theta$; in the second for-loop, it is rotated by an angle of $-\theta$. The probability that the subroutine returns the value $1$ depends on the input values \lstinline{x} and \lstinline{y}. Specifically, the probability that the subroutine returns $1$, and correspondingly that
the automaton $\mathcal{A}$ accepts 
the word $a^{{x}} b c^{{y}}$,
is $|\cos(({x}-{y}) \cdot \theta)|$.

\begin{figure}[t!]
    \centering
    \begin{subfigure}[b]{0.47\textwidth}
        \centering
        \lstset { 
            language=C,
            backgroundcolor=\color{black!5},
            basicstyle=\footnotesize\ttfamily,
            numbers=left,
            stepnumber=1,
            showstringspaces=false,
            tabsize=2,
            breaklines=true,
            xleftmargin=2em,
            xrightmargin=1em,
            framexleftmargin=1.5em,
            framexrightmargin=0em,
            columns=fullflexible 
        }
\begin{lstlisting}[mathescape]
void eq(int x, int y){
   qubit z=$\ket{0}$;
   for(int i=0; i< x; i++)
      z = rotate(z);
   for(int i=0; i< y; i++)
      z = reverse_rotate(z);
   if(measure(z)==$\ket{0}$)
      return 1;
   else return 0;
}
\end{lstlisting}
\caption{A C-like code subroutine.}
\label{Fig:c_code}
\end{subfigure}
\hfill 
\begin{subfigure}[b]{0.5\textwidth}
        \centering
        \begin{tikzpicture} [draw=black,
            node distance = 3cm, 
            on grid, 
            auto,
            every initial by arrow/.style = {thick}]

        \node (q0) [state, initial, initial text = {}] at (0,0) {$q_0$};
        \node [state] (q1) at (3,0) {$q_1$};

        \path [-stealth, thick]
             (q0) edge[] node[above,sloped] {$b$}   (q1)
             (q0) edge [loop above]  node {$a$}()
             (q1) edge [loop above]  node {$c$}()
             ;
        \end{tikzpicture}
        \caption{The corresponding QFAC $\mathcal{A}$, where the symbols $a$ and $c$ refer to the quantum operations \texttt{rotate} and \texttt{reverse\_rotate}, respectively, while the symbol $b$ refers to the identity operation.}
        \label{Fig:c_code_automaton}
\end{subfigure}
\caption{An example of a QFAC simulating the behavior of a quantum program.}
\vspace{-1.5em}
\end{figure}
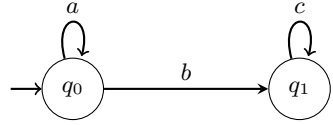
\end{example}


Analogous to classical finite automata (FAs), QFAs are categorized as either one-way \cite{Moore:00:TCS,Brodsky:02:SIAM} or two-way \cite{Kondacs:97:FOCS}. An FA is two-way if it can move its input head both backward and forward on the input tape; otherwise, it is one-way, and each input symbol can be accessed only once. For FAs, one-way and two-way models recognize exactly the class of regular languages \cite{Hopcroft:01:text}. In contrast to their classical counterparts, one-way QFAs are strictly less expressive than two-way QFAs \cite{Kondacs:97:FOCS}.

Information about the state of a quantum system is obtained through measurement, and unlike in classical systems, measurement affects the state of the system. Based on the number of measurements performed during computation, QFAs can be classified into measure-once QFAs (MO-QFAs) \cite{Moore:00:TCS} and measure-many QFAs (MM-QFAs) \cite{Brodsky:02:SIAM,Lin:12:JCSS}. In an MO-QFA, the quantum state is measured only at the end of the computation, whereas in an MM-QFA, measurements may be performed throughout the computation. In terms of expressive power, MM-QFAs are more powerful than MO-QFAs \cite{Brodsky:02:SIAM}. These one-way/two-way and measure-once/measure-many distinctions apply to the various extensions of QFAs as well.

The semantics of QFAs and their extensions vary across the literature, including cut-point semantics \cite{Blondel:05:SIAM,Bertoni:03:DLT,Mereghetti:06:RAIRO} and bounded-error semantics \cite{Brodsky:02:SIAM,Qiu:08:ICIC,Qiu:09:CoRR,Yakaryilmaz:10,Zheng:12:LA}. Regarding expressiveness, QFACs under bounded-error semantics recognize exactly the class of regular languages \cite{Qiu:09:CoRR}. In this paper, we focus on cut-point semantics, where QFAC languages are defined based on a threshold $\lambda \in [0,1]$ and a comparison operator $\bowtie \in \{>, <, \geq, \leq\}$. For instance, if $\lambda = 1/3$ and $\bowtie$ is $\geq$, the corresponding language is the set of words with an acceptance probability of at least $1/3$. We prove that QFAs and FAs are incomparable, and both are strictly less powerful than QFACs in terms of expressiveness, referring to \cref{fig:expressiveness}.

\begin{figure}[t]
    \centering
    \scalebox{0.9}{
    \begin{tikzpicture} [set/.style = {draw,opacity = 0.4,text opacity = 1}]
 
    \node[fill=cyan,ellipse, minimum width=2.5cm, minimum height=1.5cm] at (0,0) [set] {};
    \node[fill=purple,ellipse, minimum width=2.5cm, minimum height=1.5cm] at (2,0) [set] {};
    \node[fill=gray,minimum width=6cm,minimum height=2.5cm] at (1,0.2) [set] {};
    \node at (1,1.0) {$\lqfac$};
    \node[purple] at (2.2,0) {Reg};
    \node[blue] at (-0.2,0.2) {$\lbqfac$};
    \node[blue] at (-0.2,-0.2) {$= \lqfa$}; 
    \end{tikzpicture}
    }
    \caption{Relative expressiveness of quantum automata models.}
    \label{fig:expressiveness}
    \vspace{-2em}
\end{figure}
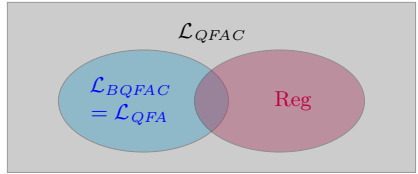

We study the emptiness problem for QFACs, which asks whether the language recognized by a QFAC is empty for a given threshold $\lambda$ and comparison operator $\bowtie$. It is known that this problem is decidable for QFAs with a strict comparison operator $\in \{>, <\}$, but undecidable with a non-strict comparison operator $\in \{\geq, \leq\}$ \cite{Blondel:05:SIAM}. We show that the same result holds for QFACs by reducing the problem to the language intersection of a QFA and an FA, which is known to be decidable in the strict case \cite{Bertoni:13:DLT}. Additionally, we study the emptiness problem for \emph{flat} QFACs, which are QFACs without nested loops, and relate it to the \emph{higher-dimensional orbit problem}, a long-standing open problem \cite{Kannan:80:STOC}. Specifically, we provide a reduction from the emptiness problem for one-loop QFACs (which are a subclass of flat QFACs) to the higher-dimensional orbit problem. Furthermore, we establish complexity lower bounds for the emptiness problem of QFACs and flat QFACs. This is achieved by reducing the DFA universality problem (PSpace-complete) and the Hamiltonian path problem (NP-complete) to these problems, respectively. The results are summarized in 
\Cref{tb:summation}.


Finally, we develop a tool for searching for witnesses to non-emptiness for one-loop QFACs. Even though the problem is decidable in the strict case, the algorithm proposed in the proof and in~\cite{Blondel:05:SIAM,Bertoni:13:DLT} requires iteratively solving a quantifier elimination problem over the reals, which does not scale in practice. Another common line of approach~\cite{Akshay:CAV:23,Tsai:CONCUR:25} is to design suitable templates for witnesses and solve them using SMT solvers such as Z3~\cite{Z3}, which also does not scale when polynomial templates and objectives are involved due to blow-ups in the number of variables. Alternatively, we propose a witness generation procedure based on \emph{nonlinear integer programming (NIP)}. Our witness does not depend on any predefined template. We systematically generate them based on the Taylor approximation of trigonometric functions and verify them using an NIP solver Mindtpy~\cite{mindtpy}. We formally prove that our witness is not only sound but also \emph{semi-complete} for the strict operator. We demonstrate through experiments that our tool can successfully verify Grover's search for up to 13 qubits and quantum random walks for up to 4 qubits.

\begin{table}[t]
\begin{tabular}{|>{\columncolor{cyan!10}}c|c|c|c|c|}
\hline
\cellcolor{gray!15}Emptiness & \cellcolor{gray!15}general & \cellcolor{gray!15}one-loop/unary QFA & \cellcolor{gray!15}flat & \cellcolor{gray!15} flat and finite group \\ \hline
strict                                 & \PSPACE-hard                            & \NP-hard                                             & \NP-hard                         &            \NP-hard          \\ \hline
non-strict                      & undecidable                     &  \makecell{decidable while $\text{rank}(P)$\\ $= 0,1,2, n-2, n-1$, $n$}                                       & ?                            & decidable             \\ \hline
\end{tabular}
\caption{Summary of results}
\label{tb:summation}
\vspace{-2em}
\end{table}

\section{Preliminaries}

We use $\mathbb{R}$ (resp. $\mathbb{N}$, $\mathbb{C}$) to denote the set of real (resp. natural, and complex) numbers, and we use $\positive$ to denote the set of positive integers. For every $s,t \in \mathbb{R}$ with $s < t$, $[s,t]$ denotes the set $\{r \in \mathbb{R} \mid s \leq r \leq t\}$.

In the Hilbert space, a \emph{quantum operation} and a \emph{quantum state} are represented by a \emph{unitary matrix} and a \emph{unit vector}, respectively, while a \emph{measurement} is associated with the application of a \emph{projection matrix}. Specifically, let $n \in \mathbb{N}_{>0}$ and let $M$ be an $n \times n$ matrix over $\mathbb{C}$. $M$ is \emph{unitary} if $M^\dagger = M^{-1}$, where $M^\dagger$ is the conjugate transpose of $M$. If all entries of $M$ are real, then $M$ is an \emph{orthogonal} matrix. We say that $M$ is a \emph{projection matrix} if $M^2 = M$. It is known that $I - M$ is a projection matrix whenever $M$ is, where $I$ denotes the identity matrix.
In this paper, we use bra-ket notation for vectors. A vector $\ket{\phi}$ is represented by a column vector, while $\bra{\phi}$ denotes its conjugate transpose and is therefore represented by a row vector. In $\mathbb{C}^n$, the vectors $\bra{1},\dots,\bra{n}$ correspond to the standard basis row vectors $(1,0,\dots,0),\dots,(0,\dots,0,1)$, respectively.
For every unitary (resp. projection) matrix $M$ and any pair of vectors $\ket{\psi}, \ket{\phi} \in \mathbb{C}^n$, one can always construct an orthogonal (resp. projection) matrix $M'$ and vectors $\ket{\psi'}, \ket{\phi'} \in \mathbb{R}^{2n}$ such that $\mybraket{\phi \mid M \mid \psi} = \mybraket{\phi' \mid M' \mid \psi'}$. Following \cite{Blondel:05:SIAM}, we assume that all matrices and vectors are over $\mathbb{R}$ for the rest of this paper.

A vector $\ket{\psi} \in \mathcal{H}^n$ is a \emph{unit vector} if $\mybraket{\psi \mid \psi} = 1$; the norm of a vector $\ket{\psi}$ is denoted as $\norm{\ket{\psi}} = \sqrt{\mybraket{\psi \mid \psi}}$. In this paper, we restrict our attention to projection matrices onto the subspace spanned by a subset of the standard basis; that is, $M = \sum_{i \in S} \ket{i}\bra{i}$ for some $S \subseteq \{1, \dots, n\}$.

\emph{Quantum finite automata with classical states} (\qfac{s}), introduced in \cite{Qiu:15:JCSS}, are an extension of deterministic finite automata (DFAs). In \cite{Qiu:15:JCSS}, during a computation of a \qfac on an input word, a sequence of quantum operations is applied to the system's ``quantum state.'' Each quantum operation is uniquely determined by the input symbol being read and the current ``classical state'' of the underlying DFA. At the end of the computation, a quantum measurement is performed based on the last reached classical state. 
In this paper, we consider both deterministic and nondeterministic \qfac{s} (where the underlying automata are DFAs or NFAs, respectively), while the quantum operations applied to the quantum states depend solely on the input symbols, rather than the classical states of the underlying automata.
The formal definition is as follows:

\begin{definition}
Given $n \in \positive$, an \emph{$n$-dimensional quantum finite automaton with classical states} is a tuple $\mathcal{A} = (\Sigma, Q, q_0, \ket{\psi_{0}}, (U_{a})_{a\in\Sigma}, (P_q)_{q\in Q}, \Delta)$,
where:
\begin{compactitem}
    \item $\Sigma$ is a finite alphabet.
    \item $Q$ is a finite set of classical states (or simply states).
    \item $q_0 \in Q$ is the initial state.
    \item $\ket{\psi_0}$ is the initial quantum state, which is a unit vector in $\mathcal{H}^n$.
    \item For each $a \in \Sigma$, $U_a$ is an $n \times n$ unitary matrix acting on $\mathcal{H}^n$.
    \item For each $q \in Q$, $P_q$ is an $n \times n$ projection matrix acting on $\mathcal{H}^n$.
    \item $\Delta \subseteq Q \times \Sigma \times Q$ is the transition relation.
\end{compactitem}
\end{definition}
A \qfac becomes a \emph{quantum finite automaton} (\qfa)~\cite{Moore:00:TCS} when it has only one classical state, that is, $|Q|=1$. 
It is \emph{deterministic} if for all $(p,a,q), (p',a',q') \in \Delta$, we have $q'=q$ whenever $(p',a')=(p,a)$. It is \emph{codeterministic} if, for every pair of transitions $(p,a,q)$ and $(p',a',q')$, we have $p=p'$ whenever $(a,q)=(a',q')$. A \qfac is said to be a \emph{bideterministic} \qfac (\bqfac) if it is both deterministic and codeterministic. It is \emph{complete} if for every state $p \in Q$ and symbol $a \in \Sigma$, there exists a state $q \in Q$ such that $(p,a,q) \in \Delta$. 


A strongly connected directed graph $G$ is \emph{flat} if it is either a simple loop or a single vertex. For a directed graph $G$, we say that $G$ is \emph{flat} if all of its strongly connected components (SCCs) are flat. A flat directed graph $G$ is a \emph{chain} if each SCC of $G$ has at most one predecessor and successor. We say that a \qfac $\mathcal{A}$ is \emph{flat} (resp. \emph{chain}) if its underlying graph is flat (resp. \emph{chain}).

A \emph{configuration} of $\mathcal{A}$ is a pair $(q,\ket{\psi}) \in Q \times \mathcal{H}^n$. The configuration $c_0=(q_0,\ket{\psi_0})$ is called the \emph{initial configuration}.  
Given a word $w=a_1\dots a_m$, a \emph{run} $\pi$ of $\mathcal{A}$ on $w$ from $p_0$ is a sequence of states $p_0\dots p_m$ such that $(p_i,a_{i+1},p_{i+1}) \in \Delta$ for each $i=0,\dots,m-1$.
The \emph{accepting probability} of $w$ along $\pi$ from quantum state $\ket{\psi}$ is:
\begin{equation}
    \val_\mathcal{A}(w,\pi,\ket{\psi}) = \norm{P_{p_m}\cdot U_{a_m}\cdots U_{a_1}\ket{\psi}}.
\end{equation}
The accepting probability of a word $w$ is the maximum value of $\val_\mathcal{A}(w,\pi,\ket{\psi_0})$ over all runs $\pi$ from $c_0$. With a slight abuse of notation, we denote this value by $\val_\mathcal{A}(w)$. The subscript is omitted when it is clear from the context.

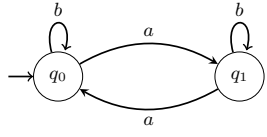
\begin{figure}[t]\centering
\scalebox{0.8}{
\begin{tikzpicture} [draw=black,
    node distance = 3cm, 
    on grid, 
    auto,
    every initial by arrow/.style = {thick}]

\node (q0) [state, 
    initial,  
    initial text = {}] {$q_0$};
\node (q1) [state,
    right = of q0] {$q_1$};

\path [-stealth, thick]
     (q0) edge[bend left] node {$a$}   (q1)
     (q1) edge[bend left] node {$a$}   (q0)
     (q0) edge [loop above]  node {$b$}()
     (q1) edge [loop above]  node {$b$}();
\end{tikzpicture}
}
\caption{\qfac $\mathcal{A}_{\text{odd}}$.}
\label{Fig:example_odd}
\end{figure}

\begin{example}\label{eg:A_odd}
Consider the 2-dimensional \qfac $\mathcal{A}_{\text{odd}}$ over $\{a,b\}$ with the initial quantum state $\ket{\psi_0} = (1,0)^\top$ as shown in \cref{Fig:example_odd}. Here, $q_0$ is the initial state, the symbol $a$ corresponds to the matrix that rotates vectors in the $xy$-plane by an angle of $1$ radian, and $b$ corresponds to the identity matrix. 

Since $\mathcal{A}_{\text{odd}}$ is deterministic, for every word $w$, there is a unique run $\pi$ on $w$ starting from $q_0$. Starting from the initial configuration $c_0$ along $\pi$, $\mathcal{A}_{\text{odd}}$ ends in a configuration $(q, \ket{\psi})$, where $q = q_1$ if and only if the number of occurrences of the symbol $a$ in $w$ is odd; otherwise, $q = q_0$. 
The resulting quantum state is $\ket{\psi} = (\cos(n), \sin(n))^\top$, where $n$ is the number of $a$'s in $w$.
By setting the projection matrix $P_{q_0}$ to be the null matrix and the projection matrix $P_{q_1}$ to be the matrix projecting onto the $x$-axis, we have:
$\val_{\mathcal{A}_{\text{odd}}}(w) = \cos^2(n)$ if the number of  $a$'s in $w$ is odd, and 0 otherwise.
\qed
\end{example}

Analogous to probabilistic finite automata, \qfac{s} possess both \emph{cut-point semantics} and \emph{bounded-error semantics}. In this paper, we focus on the former. For instance, given $\lambda \in [0,1]$, $L_{> \lambda}(\mathcal{A})$ is the language consisting of all words $w$ such that $\val(w) > \lambda$. The languages $L_{< \lambda}(\mathcal{A})$, $L_{\geq \lambda}(\mathcal{A})$, $L_{\leq \lambda}(\mathcal{A})$, and $L_{=\lambda}(\mathcal{A})$ are defined analogously. 
For ease of exposition, we use $\cmp$ to denote the set of comparison operators $\{>, <, \geq, \leq\}$. Let ${\bowtie} \in \cmp$. We use $\Bar{\bowtie}$ and $\Tilde{\bowtie}$ to denote the operators obtained by reversing the comparison and by reversing only the strictness of the inequality, respectively. For example, if $\bowtie$ is $\geq$, then $\Bar{\bowtie}$ is $<$ and $\Tilde{\bowtie}$ is $\leq$.
Recall that an NFA can always be turned into an equivalent one that is complete by delegating non-existing edges to a rejecting sink state. A similar technique applies to \qfac{s} but requires modifications for some combinations of $\bowtie$ and $\lambda$, as detailed in \Cref{app:complete}. Unless otherwise stated, all \qfac{s} in the remainder of this paper are assumed to be complete.


The \emph{emptiness problem} for a given $\mathcal{A}$, ${\bowtie} \in \cmp$, and $\lambda \in [0,1]$ is to determine whether $L_{\bowtie \lambda}(\mathcal{A}) = \emptyset$, while the \emph{universality problem} is to determine whether $L_{\bowtie \lambda}(\mathcal{A}) = \Sigma^*$. By definition, we have $L_{\bowtie \lambda}(\mathcal{A}) = \emptyset$ if and only if $L_{\bar{\bowtie} \lambda}(\mathcal{A}) = \Sigma^*$.

We use $\fixlanclass{\bowtie}{\lambda}$ to denote the class of languages $L_{\bowtie \lambda}(\mathcal{A})$ for all \qfac{s} $\mathcal{A}$, and let $\lanclass{\bowtie} = \bigcup_{\lambda \in [0,1]} \fixlanclass{\bowtie}{\lambda}$ and $\lqfac = \bigcup_{\bowtie \in\cmp} \lanclass{\bowtie}$. The class $\lqfac$ becomes $\lqfa$ (resp., $\lbqfac$)  when restricted to \qfa{s} (resp., \bqfac{s}). Due to the completeness assumption, there are four particular classes of language that are always trivial, namely $\fixlanclass{>}{1}=\fixlanclass{<}{0}=\{\emptyset\}$ and $\fixlanclass{\geq}{0}=\fixlanclass{\leq}{1}=\{\Sigma^{\ast}\}$. A pair $(\bowtie,\lambda)\in \cmp\times[0,1]$ is called \emph{trivial} if it corresponds to one of these four classes.



\hide{
\section{closure properties}

\begin{lemma}\label{lem:close-under-union}
$\lanclass{\textit{\qfac}}$ is closed under union.
\end{lemma}
\begin{proof}
Let $\mathcal{A}_i = (\Sigma, Q_i, q_{0i}, \psi_{0i}, (U^{(i)}_{a})_{a\in\Sigma}, (P^{(i)}_q)_{q\in Q_i}, \Delta_i)$, $i=1,2$, be two \qfac{s}. We define $\mathcal{A} =(\Sigma, Q, q_0, \psi_0, (U_a)_{a \in \Sigma}, (P_{q,p})_{(q,p) \in Q}, \Delta)$ via the standard product construction, that is, 
\begin{itemize}
    \item $Q = Q_1 \times Q_2$;
    \item $q_0 = (q_{01}, q_{02})$;
    \item $\psi_0 = \psi_{01} \otimes \psi_{02}$;
    \item $U_a = U^{(1)}_a \otimes U^{(2)}_a$ for each $a \in \Sigma$;
    \item $P_{(q,p)}= P^{(1)}_q \otimes P^{(2)}_p$ for each $(q,p) \in Q$; and $((q,p),a,(q'p')) \in \Delta$ whenever $(q,a,q') \in \Delta_1$ and $(p,a,p') \in \Delta_2$.
\end{itemize}
not possible
\end{proof}

\jal{not closed even under union}
}

\section{Expressiveness}\label{sec:expressiveness}

In this section, we study the expressive power of \qfac{s} and \qfa{s}. We show that the class of regular languages is a proper subclass of the languages recognized by \qfac{s} with both strict and non-strict thresholds, and that $\lqfa$ are incomparable with regular languages. See \Cref{fig:expressiveness}.

We begin with a basic closure property of $\lqfac$ over complementation. Let ${\cal A}$ be a \qfac. Recall that a word $w$ is in the language $L_{\bowtie\lambda}({\cal A})$ if and only if $\val_{\cal A}(w)\bowtie\lambda$, which also implies $val_{\cal A}(w)\Bar{\bowtie}\lambda$ does not hold, hence $w\not\in L_{\Bar{\bowtie}\lambda}({\cal A})$. This gives us the following lemma.

\begin{restatable}{lemma}{lemlanguagesymmetric}\label{lem:language_symmetric}
Let $\mathcal{A}$ be a \qfac over $\Sigma$. For every $\lambda\in [0,1]$ and ${\bowtie}\in\cmp$, $L_{\bowtie \lambda}(\mathcal{A}) = \Sigma^* \setminus L_{\Bar{\bowtie}\lambda}(\mathcal{A})$.
\end{restatable}

A language is regular iff it is recognized by a DFA. Suppose $\mathcal{A}$ is a DFA over $\Sigma$. We can construct a 2-dimensional \qfac $\mathcal{B}$ such that the underlying automaton of $\mathcal{B}$ is $\mathcal{A}$ and every symbol $a \in \Sigma$ corresponds to the identity matrix. Excluding trivial cases, for any ${\bowtie} \in \cmp$ and $\lambda \in [0,1]$, by appropriately choosing the initial quantum state $\ket{\psi_0}$ and the projection matrices $P_q$, we can ensure that $\val_\mathcal{B}(w) \bowtie \lambda$ iff $w$ is accepted by $\mathcal{A}$. For instance, if $\bowtie$ is $\geq$, we can set $\ket{\psi_0}=(1,0)^\top$ and let $P_q$ project onto the first component if $q$ is an accepting state in $\mathcal{A}$, and be the null matrix otherwise. Note that for cases where ${\bowtie} \in \{<, \leq\}$, the choice of the initial quantum state depends on $\lambda$. Specifically, we set $\ket{\psi_0}=(\sqrt{\lambda/2},\sqrt{(2-\lambda)/2})^\top$. Therefore, we have:

\begin{restatable}{lemma}{lemQFACsubRegular}\label{lem:QFAC_subsumbs_regular}
    For every non-trivial pair of $(\bowtie,\lambda) \in \cmp\times [0,1]$, $\fixlanclass{\bowtie}{\lambda}$ subsumes the class of regular languages.
\end{restatable}

It is a natural question to ask whether the inclusion in \Cref{lem:QFAC_subsumbs_regular} is strict and whether this strictness depends on the choice of $\bowtie$ and $\lambda$. This question was partially answered in~\cite{Brodsky:02:SIAM} by showing that there exists a 2-dimensional \qfa $\mathcal{A}$ over $\Sigma=\{a,b\}$ such that $L_{>0}({\cal A})$ is exactly the non-regular language
$L_{\neq} = \{ w \in \Sigma^* : |w|_a \neq |w|_b \}$. Specifically, consider the rotation matrix $M_a$ corresponding to the symbol $a$ in \Cref{eg:A_odd}, and let $M_b$ be its inverse $M_a^{-1}$. Since $\cos(1)$ and $\sin(1)$ are transcendental numbers \cite{Baker:75:text}, it follows that $M_{\sigma_1} M_{\sigma_2} \cdots M_{\sigma_n} = I$ (where $w = \sigma_1 \cdots \sigma_n \in \Sigma^*$) if and only if the number of $a$'s is equal to the number of $b$'s in $w$. By setting the initial quantum state to $\ket{\psi_0}=(1,0)^\top$ and letting $P$ be the projection onto the second component, we obtain $\val(w) > 0$ if and only if $w \in L_{\neq}$. Recall that a \qfa has only one classical state and is deterministic. Similar arguments can be applied to construct automata for other values of $\lambda$ and also the case of $<$, that accept non-regular languages. For the non-strict case, we can use~\cref{lem:language_symmetric} to show that they accept $\Sigma^{\ast}\setminus L_{\neq}$, which is also non-regular. Hence, we have:

\begin{restatable}{lemma}{lemQFAnonRegular}\label{lem:QFA_nonRegular}
    For every non-trivial pair of $(\bowtie, \lambda) \in \cmp \times [0,1]$, there exists a \qfa $\mathcal{A}$ such that $L_{\bowtie \lambda}(\mathcal{A})$ is not regular.
\end{restatable}

From \Cref{lem:QFAC_subsumbs_regular} and \Cref{lem:QFA_nonRegular}, we immediately obtain:

\begin{theorem}\label{thm:qfacs_subsumes_regular}
The class of regular languages is a proper subclass of $\fixlanclass{\bowtie}{\lambda}$ for every non-trivial pair $(\bowtie, \lambda) \in \cmp \times [0,1]$.
\end{theorem}

To show that $\lanclass{\textit{\qfa}}$ is incomparable with the class of regular languages, we must also show that regular languages are not subsumed by $\lanclass{\textit{\qfa}}$. In particular, we provide a concrete instance of a regular language that is not recognizable by any \qfa. We begin by recalling the pumping lemma for \qfa{s} proven in~\cite{Moore:00:TCS}.

\begin{lemma}[\cite{Moore:00:TCS}]\label{lem:qfa_pumping}
    Let $\mathcal{A}$ be a \qfa over $\Sigma$. Then for every word $w\in\Sigma^*$ and every $\varepsilon > 0$, there exists $k\in \mathbb{N}_{>0}$ such that for all $u,v\in \Sigma^*$,
    $|\val_{\mathcal{A}}(u(w)^kv)-\val_{\mathcal{A}}(uv)| \leq \varepsilon$.
\end{lemma}

Let $L_{\neg(bb)}$ be the language over $\Sigma=\{a,b\}$ consisting of all words $w$ that do not contain two consecutive $b$'s. It is easy to see that this language is regular. Using \Cref{lem:qfa_pumping}, it can be shown that $L_{\neg(bb)}$ is not in $\lanclass{\textit{\qfa}}$ (see \cite{Moore:00:TCS} or \Cref{app:expressiveness} for details). Thus, we have:

\begin{restatable}{lemma}{lemNonQFArecognizable}\label{lem:non-QFA_recognizable}
    $L_{\neg (bb)}$ is not in $\lqfa$.
\end{restatable}

As a result, we obtain the following theorem:

\begin{theorem}\label{thm:qfas_incomparable_with_regular}
    $\lanclass{\textit{\qfa}}$ and the class of regular languages are incomparable.
\end{theorem}

Recall the pumping lemma for NFAs: let $\mathcal{A}$ be an NFA. There exists some number $n\in\mathbb{N}$ such that for every word $w \in L(\mathcal{A})$ with $|w| \geq n$, $w$ can be partitioned into three substrings $w_{1},w_{2},w_{3}$ such that $w = w_1 w_2 w_3$, $|w_2| > 0$, $|w_1 w_2| \leq n$, and $w_1(w_2)^i w_3 \in L(\mathcal{A})$ for all $i \in \mathbb{N}$. 
In the pumping lemma for \qfa{s}, the entire word $w$ is pumped, whereas for NFAs, the word is partitioned into $w_1, w_2, w_3$ and only the factor $w_2$ is pumped. This difference stems from the fact that \qfa{s} possess only a single classical state; thus, any run on a word can always be pumped. In contrast, a run in an NFA can be pumped only when it contains a cycle. By combining the pumping lemma for \qfa{s} (\Cref{lem:qfa_pumping}) with the one for NFAs, we can derive a pumping lemma for \qfac{s}.

\begin{restatable}{lemma}{lemPumpingQFAC}\label{lem:pumping2}
    Let $\mathcal{A}$ be a \qfac over $\Sigma$, there exists $n \in \mathbb{N}$ with $n>1$ such that for every $(\bowtie,\lambda)\in\cmp\times [0,1]$ and every word $w\in L_{\bowtie\lambda}({\cal A})$ with $\abs{w}\geq n$, there exists a partition $w_{1}w_{2}w_{3}$ of $w$ such that 
    \begin{compactitem}
        \item $\abs{w_{2}}>0$, $\abs{w_{1}w_{2}}< n$,
        \item $\forall\varepsilon>0$, $\exists k\in\mathbb{N}_{>0},\,\abs{\val_{\cal A}(w)-\val_{\cal A}(w_{1}w_{2}^{k}w_{3})}\leq\varepsilon$.
    \end{compactitem}
\end{restatable}
Using~\Cref{lem:pumping2}, we can show that $L_{=} = \Sigma^{\ast}\setminus L_{\neq} = \{w\in\Sigma^{\ast} : \abs{w}_{a} = \abs{w}_{b}\}$ is not recognizable by a QFAC with strict $\bowtie$ for any $\lambda$. To see this, let ${\bowtie}\in\{>,<\}$ and assume towards a contradiction that there exists a \qfac ${\cal A}$ such that $L_{\bowtie\lambda}({\cal A}) = L_{=}$ for some $\lambda\in [0,1]$. Let $n$ be the constant given by~\Cref{lem:pumping2}, and let $w = a^{n+1}b^{n+1}$. Clearly $w\in L_{=}$ and thus $w\in L_{\bowtie\lambda}({\cal A})$. By~\Cref{lem:pumping2}, there exists a partition $w = w_{1}w_{2}w_{3}$ such that for all $\varepsilon>0$, there exists some $k>0$ such that $\abs{\val_{\cal A}(w)-\val_{\cal A}(w_{1}w_{2}^{k}w_{3})}\leq\varepsilon$. In particular, for $\varepsilon = \frac{1}{2}\abs{\lambda - \val_{\cal A}(w)}$ it holds that $\val_{\cal A}(w_{1}w_{2}^{k}w_{3})\bowtie\lambda$ and thus $w_{1}w_{2}^{k}w_{3}\in L_{\bowtie\lambda}({\cal A})$. However, since $\abs{w_{1}w_{2}}< n$, $w_{2}$ must contain only the letter $a$. As $w_{2}$ cannot be empty, $w_{1}w_{2}^{k}w_{3}\notin L_{=}$,contradicting our initial assumption. By~\Cref{lem:language_symmetric}, there is no \qfac that recognizes the language $L_{\neq}$ with ${\bowtie}\in\{\geq,\leq\}$ for any $\lambda$ either. This gives the following theorem.

\begin{theorem}
    $\lanclass{\bowtie}$ and $\lanclass{\bowtie'}$ are incomparable for ${\bowtie} \in \{>,<\}$ and ${\bowtie'}\in\{\geq,\leq\}$.
\end{theorem}

By \Cref{thm:qfacs_subsumes_regular} and \Cref{thm:qfas_incomparable_with_regular}, it follows that $\lanclass{\textit{\qfa}}$ is strictly subsumed by $\lanclass{\textit{\qfac}}$. The following result suggests that if we restrict the underlying automata of \qfac{s} to be bideterministic, then $\lanclass{\textit{\qfac}}$ collapses to $\lanclass{\textit{\qfa}}$.

\begin{restatable}{lemma}{lemBideterministicQFACtoQFA}\label{lem:bideterministicQFAC_to_QFA}
    Let $\mathcal{A}$ be a \qfac whose underlying automaton is bideterministic. Then, a \qfa $\mathcal{B}$ satisfying $L_{\bowtie \lambda}(\mathcal{B}) = L_{\bowtie \lambda}(\mathcal{A})$ for every $(\bowtie, \lambda) \in \cmp \times [0,1]$ can be effectively obtained.
\end{restatable}

The intuition behind the proof of \Cref{lem:bideterministicQFAC_to_QFA} is as follows. When the underlying automaton $\mathcal{A}_{\text{under}}$ of $\mathcal{A}$ is bideterministic, for each symbol $a \in \Sigma$, the transitions of $\mathcal{A}_{\text{under}}$ form a permutation on its states. Since permutations are reversible, they can be represented as quantum operations; thus, the computation of $\mathcal{A}_{\text{under}}$ can be simulated by a \qfa. The detailed proof is provided in \Cref{app:expressiveness}.

\section{Decidability and Complexity of Emptiness}\label{sec:decidability_complexity}
In this section, we study the decidability and complexity of the emptiness problem for \qfac{s}. The section is organized as follows. We first show in~\Cref{sec:deci_qfac} that the emptiness problem for \qfac{s} is as hard as the same problem for \qfa{s} in general, which is shown to be undecidable for the non-strict operator and decidable for the strict operator in~\cite{Blondel:05:SIAM}. Then in~\Cref{sec:deci_flat}, we investigate a few decidable fragments of the undecidable emptiness problem, i.e., the non-strict operator case, by imposing the \emph{flatness} assumption on the underlying structure of the \qfac. Finally, in~\Cref{sec:complexity}, we establish complexity lower bounds for \qfac{s} and flat \qfac{s}.

\subsection{Decidability for \qfac{s}}\label{sec:deci_qfac}
We start with showing the decidability result of the emptiness problem for \qfac{s}. The undecidable result with the non-strict operators is trivial since \qfa{s} are a special case of \qfac{s}. For the decidable result, we show that we can reduce the emptiness problem of \qfac{s} to the language intersection problem between a \qfa and a finite automata, which is a special case of the result in~\cite{Bertoni:13:DLT}.

Let $\mathcal{A} = (\Sigma, Q, q_0, \ket{\psi_{0}}, (U_{a})_{a\in\Sigma}, (P_q)_{q\in Q}, \Delta)$ be a \qfac. Recall that a word $w\in\Sigma^{\ast}$ is $(\bowtie,\lambda)$-accepted by ${\cal A}$ if $\max_{\pi\in Path_{\cal A}(w)}\val_{\cal A}(w,\pi,\ket{\psi_{0}})\bowtie\lambda$. When $\bowtie$ is $>$, this is equivalent to the satisfiability of $\exists \pi\in Path_{\cal A}(w), \val_{\cal A}(w,\pi,\ket{\psi_{0}})>\lambda$. When $\bowtie$ is $<$, it becomes $\forall \pi\in Path_{\cal A}(w), \val_{\cal A}(w,\pi,\ket{\psi_{0}})<\lambda$. The asymmetry in the quantifier suggests that we shall handle them separately. We first show the decidability for the case of $>$, and then show how to extend the approach to the case of $<$.

In the following, we denote by $L_{>\lambda}({\cal A})[q]$ the set of all words $w\in\Sigma^{\ast}$ such that $w$ induces a run from $q_{0}$ to $q$ and satisfies $\norm{P_{q}U_{w}\ket{\psi_{0}}}>\lambda$. By the definition of the language $L_{>\lambda}({\cal A})$, any accepting word $w$ must belong to some $L_{>\lambda}({\cal A})[q]$ for some $q\in Q$, and vice versa. Therefore we have the following decomposition of languages:
\begin{equation}\label{eq:decomposition_gt}
    L_{>\lambda}({\cal A}) = \bigcup_{q\in Q} L_{>\lambda}({\cal A})[q].
\end{equation}
Since the set $Q$ is finite, it suffices for us to show that the emptiness of $L_{>\lambda}({\cal A})[q]$ is decidable for any $q\in Q$.

The emptiness problem for $L_{>\lambda}({\cal A})[q]$ can be reduced to the intersection between a \qfa and an NFA.
Specifically, let ${\cal A}_{q}$ be the \qfa obtained by removing all classical components of ${\cal A}$ and fixing the projection operator to $P_{q}$. Furthermore, let ${\cal B}_{q}$ be the underlying NFA of ${\cal A}$ with $q$ as the unique accepting state. From these constructions, it follows that $L_{>\lambda}({\cal A})[q] = L_{>\lambda}({\cal A}_{q})\cap L({\cal B}_{q})$, yielding the following lemma:

\begin{restatable}{lemma}{lemNFAQFAintersect}\label{lem:NFA_QFA_Intersection}
    Let $\mathcal{A} = (\Sigma, Q, q_0, \ket{\psi_{0}}, (U_{a})_{a\in\Sigma}, (P_q)_{q\in Q}, \Delta)$ be a \qfac. For any state $q\in Q$, there exists a \qfa ${\cal A}_{q}$ and a NFA ${\cal B}_{q}$ defined over $\Sigma$ such that for any $\lambda\in [0,1]$, $L_{>\lambda}({\cal A})[q] = L_{>\lambda}({\cal A}_{q})\cap L({\cal B}_{q}).$
\end{restatable}

By the result in~\cite{Bertoni:13:DLT}, the emptiness problem of the intersection of a \qfa with strict cutpoint semantics and a linear context-free language is decidable. A context-free grammar is \emph{linear} if every production rule has at most one nonterminal on the right-hand side \cite{Hopcroft:01:text}. Since regular languages are linear context-free~\cite{Chomsky:58:IC}, combining the above lemma with the result in~\cite{Bertoni:13:DLT} yields decidability of the emptiness of $L_{>\lambda}({\cal A})[q]$ for any $q\in Q$ and $\lambda\in [0,1]$. Hence, we have the following result.

\begin{restatable}{theorem}{thmEmptDeci}\label{thm:decidability_emptiness}
    For a \qfac $\mathcal{A}$, the emptiness of $L_{>\lambda}({\cal A})$ is decidable for $\lambda \in [0,1]$.
\end{restatable}

Now we extend this result to the case of $<$. Note that the decomposition in~\cref{eq:decomposition_gt} no longer holds in this case so the same reduction will not work. In this semantics, we need to check the condition $\norm{P_{q}U_{w}\ket{\psi_{0}}}<\lambda$ for all $q$ reachable by $w$, which is also the subset of states $w$ would reach in the DFA obtained by the subset construction. This leads us to consider a decomposition based on subset of states instead of a single state, which we present below.

Fix any non-empty subset $S\subseteq Q$, and let ${\cal B}[S]$ be the DFA obtained via determinizing the classical part of ${\cal A}$ and fixing $S$ as the unique accepting state. Any word $w=a_{1}\dots a_{n}\in L({\cal B}[S])$ induces a unique run $S_{0}S_{1}\dots S_{n}$ on $w$ of ${\cal B}$ such that $S_{0} = \{q_{0}\}$ and $S_{n} = S$. By the property of the subset construction, any run $q_{0}q_{1}\dots q_{n}$ on $w$ of ${\cal A}$ must have $q_{i}\in S_{i}$ for all $i$. For the word $w$ to be in the language $L_{<\lambda}({\cal A})$, we require $\norm{P_{q} U_{w}\ket{\psi_{0}}} < \lambda$ for all $q\in S$. The latter part can be captured by a set of \qfa{s} ${\cal A}_{q}$ for each $q\in S$, thus suggesting the following decomposition of the language:
\begin{equation}\label{eq:decomposition_lt}
    L_{<\lambda}({\cal A}) = \bigcup_{S\subseteq Q} \left(\bigcap_{q\in S} L_{<\lambda}({\cal A}_{q})\cap L({\cal B}[S])\right).
\end{equation}

By~\cref{eq:decomposition_lt}, the problem reduces to checking the emptiness of the intersection of a DFA and finitely many \qfa{s}. In particular, these \qfa{s} share the same initial state and unitary operators. Below, we briefly outline how to extend the decidability result in~\cite{Bertoni:13:DLT} to this case. 

Let $({\cal A}_{i} = (\Sigma,\ket{\psi_{0}},(U_{a})_{a\in\Sigma},P_{i}))_{i=1,\dots, K}$ be a finite set of \qfa{s} and let ${\cal G}$ be a linear context-free grammar. Furthermore, we denote $L({\cal G})$ as the language induced by ${\cal G}$ and $\Gamma({\cal G})$ as the set of matrices $U_{w}$ induced by words $w\in L({\cal G})$. Similar to~\cite{Blondel:05:SIAM,Bertoni:13:DLT}, the emptiness problem of the language $\bigcap_{i=1}^{K}L_{<\lambda}({\cal A}_{i})\cap L({\cal G})$ is equivalent to the satisfiability of the following formula:
\begin{equation}\label{eq:emptiness_formula}
    \forall M\in \mathbb{R}^{n\times n}, M\in\Gamma({\cal G})\implies \bigvee_{i=1}^{K}\norm{P_{i}M\ket{\psi_{0}}}\geq\lambda.
\end{equation}
We note that the function $f_{i}\colon M\mapsto \norm{P_{i}M\ket{\psi_{0}}}$ is continuous and thus the set of $M$ satisfying the right-hand side of~\cref{eq:emptiness_formula} forms a closed set. Therefore, $\Gamma({\cal G})$ can be equivalently replaced with its topological closure. In~\cite{Bertoni:13:DLT},  they show that the closure $\Bar{\Gamma({\cal G})}$ is semi-algebraic and effectively constructible. The rest follows from the decidability of the first-order theory of real closed fields.

\begin{restatable}{theorem}{IntersecMultiQFA}\label{thm:intersection_multi_qfa}
    Let $L({\cal G})$ be a linear context-free language and ${\cal A}_{1},\dots,{\cal A}_{K}$ be K distinct \qfa{s} that share the same set of unitary operators and initial state. For any $\lambda\in [0,1]$, it is decidable whether $\bigcap_{i=1}^{K}L_{<\lambda}({\cal A}_{i})\cap L({\cal G})$ is empty.
\end{restatable}

Finally, since $Q$ is finite, we have finitely many non-empty subsets $S\subseteq Q$. By combining~\cref{eq:decomposition_lt} and~\Cref{thm:intersection_multi_qfa}, we can decide the emptiness of $L_{<\lambda}({\cal A})$ by checking the emptiness of $\bigcap_{q\in S} L_{<\lambda}({\cal A}_{q})\cap L({\cal B}[S])$ for each non-empty subset $S\subseteq Q$. We conclude this section with the following theorem.

\begin{restatable}{theorem}{thmEmptDeciFull}\label{thm:decidability_emptiness_full}
    For a \qfac $\mathcal{A}$ and a non-trivial pair of $(\bowtie,\lambda)\in\cmp\times [0,1]$, the emptiness of $L_{\bowtie\lambda}({\cal A})$ is decidable if ${\bowtie} \in \{>, <\}$ and undecidable if ${\bowtie} \in \{\geq, \leq\}$.
\end{restatable}

\subsection{Decidability for Flat \qfac{s}}\label{sec:deci_flat}

By \Cref{thm:decidability_emptiness_full}, the emptiness problem for \qfac{s} under non-strict comparison operators is undecidable. This raises the question of which structural restrictions might render the problem decidable for non-strict cases. In this subsection, we investigate the emptiness problem for flat \qfac{s} (those subject to a specific structural constraint) and relate the emptiness problem for ``one-loop'' \qfac{s}, a subclass of flat \qfac{s}, to the \emph{higher-dimensional orbit problem} \cite{Kannan:80:STOC}.

By the standard decomposition of automata, we have:

\begin{restatable}{lemma}{lemFlattoChain}\label{lem:flat2chain}
    Let $\mathcal{A}$ be a flat \qfac. We can effectively decompose it into a set of chain \qfac{s} $\mathcal{A}_1, \dots, \mathcal{A}_m$ such that $L_{\bowtie \lambda}(\mathcal{A}) = \bigcup_{i=1}^{m}L_{\bowtie \lambda}(\mathcal{A}_i)$.
\end{restatable}

Based on \cref{lem:flat2chain}, we assume that all flat \qfac{s} are chains when considering the emptiness problems in the remainder of this paper.

Let $\mathcal{A}$ be a chain \qfac and $\lambda \in [0,1]$. Consider the non-strict comparison operator $\geq$; then $L_{\geq \lambda}(\mathcal{A}) \neq \emptyset$ if and only if $L_{> \lambda}(\mathcal{A}) \neq \emptyset$ or $L_{= \lambda}(\mathcal{A}) \neq \emptyset$. By \cref{thm:decidability_emptiness_full}, determining whether $L_{> \lambda}(\mathcal{A})$ is empty is decidable. The analysis for the comparison operator $\leq$ is analogous. Accordingly, for ${\bowtie} \in \{\geq, \leq\}$, the decidability of the emptiness problem for $L_{\bowtie \lambda}(\mathcal{A})$ reduces to the decidability of $L_{= \lambda}(\mathcal{A}) \neq \emptyset$.

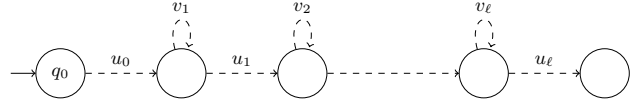
\begin{figure}[t]\centering
\scalebox{0.8}{
\begin{tikzpicture} [draw=black,
    node distance = 2cm, 
    on grid, 
    auto]

\node (q0) [state, 
    initial,  
    initial text = {}] {$q_0$};
\node (q1) [state,
    right = of q0] {};
\node (q2) [state,
    right = of q1] {};
\node (q3) [state] at (7,0) {};
\node (q4) [state,
    right = of q3] {};
    
\path [->, dashed]
     (q0) edge[] node {$u_0$}   (q1)
     (q1) edge [loop above]  node {$v_1$}()
     (q1) edge[] node {$u_1$}   (q2)
     (q2) edge [loop above]  node {$v_2$}()
     (q2) edge[] node {}   (q3)
     (q3) edge[] node {$u_\ell$}   (q4)
     (q3) edge [loop above]  node {$v_\ell$}()
     ;
\end{tikzpicture}
}
\caption{Chain \qfac $\mathcal{A}$ with $\ell$ loops.}
\label{Fig:chain_qfac}
\vspace{-1.5em}
\end{figure}

Suppose $\mathcal{A}$ has $\ell$ loops (see \cref{Fig:chain_qfac}). Using the same technique as in the proof of \Cref{lem:flat2chain}, we can assume without loss of generality that $q_\ell$ is the only state of $\mathcal{A}$ for which the projection matrix $P_{q_\ell}$ is non-null; furthermore, the entering and leaving states of every loop are identical. It follows that there exists $w \in \Sigma^*$ such that $\val_\mathcal{A}(w) = \lambda$ if and only if there exists an accepting run $\pi$ from $q_0$ to $q_\ell$ such that $\val_\mathcal{A}(w, \pi, \ket{\psi_0}) = \lambda$. Accordingly, $w \in L_{= \lambda}(\mathcal{A})$ if and only if there exist indices $i_1, \dots, i_\ell \in \mathbb{N}$ and words $u_0, \dots, u_\ell$ and $v_1, \dots, v_\ell$ such that $w = u_0(v_1)^{i_1}u_1 \dots (v_\ell)^{i_\ell}u_\ell$, where each $u_{i-1}$ corresponds to a simple path $\pi_i$ from $q_{i-1}$ to $q_i$, and each $v_j$ corresponds to a simple loop $\rho_j$ from $q_j$. In other words, such a word $w$ exists if and only if the following formula is satisfiable:
\begin{equation}\label{eq:flat_ell_loops}
    \exists i_1,\dots,i_\ell\in\mathbb{N}.\norm{P_{q_\ell} U_{u_\ell} (U_{v_\ell})^{i_\ell} \cdots U_{u_1} (U_{v_1})^{i_1} U_{u_0} \ket{\psi_0}} = \lambda.
\end{equation}

Given a unitary matrix $M$, if $M$ has \emph{finite order} (i.e., $M^i = I$ for some $i \in \mathbb{N}_{>0}$), then $S = \{M^i \mid i \in \mathbb{N}\}$ is a finite set. This set $S$ can be effectively derived when the entries of $M$ are algebraic \cite{Beals:99:JCSS}. Accordingly, if each $U_{v_j}$ has finite order, then the satisfiability of \cref{eq:flat_ell_loops} is decidable, and we have the following:

\begin{lemma}
    Let $\mathcal{A}$ be a chain \qfac and $\lambda \in [0,1]$. If the unitary matrix corresponding to each simple loop of $\mathcal{A}$ has finite order, then it is decidable whether $L_{= \lambda}(\mathcal{A}) = \emptyset$.
\end{lemma}

Next, we relate the emptiness problem for ``one-loop'' \qfac{s} to the \emph{higher-dimensional orbit problem}: given a matrix $M$, a vector $\ket{\psi}$, and a subspace $W$, decide whether there exists $i \in \mathbb{N}$ such that $M^i \ket{\psi} \in W$. This problem remains open when $M$, $\ket{\psi}$, and $W$ are defined over rational numbers and the dimension of $W$ is greater than 4 \cite{Chonev:13:STOC}.

Consider \cref{eq:flat_ell_loops} and assume that $\mathcal{A}$ has only one loop in an $n$-dimensional vector space. Since $U\ket{\psi}$ remains a unit vector for any unitary matrix $U$ and unit vector $\ket{\psi}$, deciding whether \cref{eq:flat_ell_loops} is satisfiable becomes a matter of determining the satisfiability of the following formula for given unitary matrices $U$ and $V$, a projection matrix $P$, a unit vector $\ket{\psi}$, and $\lambda \in [0,1]$:
\begin{equation}\label{eq:one_loop}
    \exists i \in \mathbb{N} : \norm{P U V^i \ket{\psi}} = \lambda.
\end{equation}

Recall the following basic properties from linear algebra: (L1) for a vector space $\mathcal{S}$ and a projection matrix $P$, the range $\{P\ket{\phi} \mid \ket{\phi} \in \mathcal{S}\}$ forms a subspace; (L2) a unitary matrix maps a subspace to another subspace of the same dimension. Combining (L1) and (L2), we observe that the operator $PU$ in \cref{eq:one_loop} effectively restricts the measurement to a subspace $W$. Suppose that $\lambda=1$. Since $V^i\ket{\psi}$ is always a unit vector, \cref{eq:one_loop} is satisfiable if and only if there exists $i \in \mathbb{N}$ such that $V^i \ket{\psi} \in W$, which is an instance of the higher-dimensional orbit problem. 
In general, when $\lambda=1$, the emptiness problem for \qfac{s} in the non-strict case is undecidable, even for \qfa{s} \cite{Blondel:05:SIAM}. It remains unclear whether this problem is decidable for flat \qfac{s} or specifically for one-loop \qfac{s}. However, based on the above analysis, we have:

\begin{lemma}
    If the emptiness problem for one-loop \qfac{s} in the non-strict case is undecidable when $\lambda=1$, then the higher-dimensional orbit problem is also undecidable.
\end{lemma}

On the other hand, when $0 < \lambda < 1$, the set $W$ of vectors $\ket{\psi}$ satisfying \cref{eq:one_loop} does not form a subspace. Instead, when the rank of $P$ is $1$, $W$ is the union of two affine spaces. The reason is as follows:

By the previously mentioned properties (L1) and (L2), the rank of the operator $PU$ in \cref{eq:one_loop} is equal to the rank of $P$. For ease of exposition, assume $PU = P$ and $P = \ket{1}\bra{1}$. For every vector $\ket{\psi}$ in an $n$-dimensional space, we can write $\ket{\psi} = a_1\ket{1} + \dots + a_n \ket{n}$ for some $a_1, \dots, a_n \in \mathbb{R}$. Consequently, $\ket{\psi}$ satisfies $\norm{P\ket{\psi}} = \lambda$ if and only if $a_1 = \lambda$ or $a_1 = -\lambda$, which defines two affine spaces. The ``affine'' orbit problem can be reduced to the original higher-dimensional orbit problem by increasing the dimension of the space by 1 \cite{Chonev:16:JACM}. In \cite{Chonev:13:STOC}, it is shown that when the dimension of the target subspace is 1, 2, or 3, the higher-dimensional orbit problem is decidable. Specifically, the complexity is polynomial when the dimension is 1, and $\complexityFont{NP}^{\complexityFont{RP}}$ when the dimension is 2 or 3, independent of the dimension of the ambient vector space.

Consider \cref{eq:one_loop}. When the rank of $P$ is 0, the formula is clearly never satisfiable for $\lambda > 0$. If the rank of $P$ is 1, the rank of $I-P$ is $n-1$, and vice versa. Furthermore, \cref{eq:one_loop} is satisfiable if and only if the formula:
\[\exists i \in \mathbb{N}. \norm{(I-P)U V^i \ket{\psi}} = 1 - \lambda\] 
is satisfiable. Together with \cref{lem:language_symmetric}, we obtain the following result:

\begin{lemma}\label{lem:one-loop}
Let $n \in \mathbb{N}_{>0}$. With respect to \cref{eq:one_loop}, the emptiness of one-loop \qfac{s} in $\mathcal{H}^n$ is decidable when the rank of $P$ is 0, 1, 2, $n-2$, $n-1$, or $n$.
\end{lemma}

Recall that a \qfa is a \qfac with one state. Therefore, a unary \qfa (where $\Sigma$ is a singleton) is a flat \qfac. By \Cref{lem:one-loop}, we obtain the following:

\begin{corollary}
    Given $n \in \positive$, the non-strict emptiness of a unary \qfa in $\mathcal{H}^n$ is decidable when the rank of the projection matrix is 0, 1, 2, $n-2$, $n-1$, or $n$.
\end{corollary}

\subsection{Complexity Lower Bounds}\label{sec:complexity}

In the rest of this section, we provide complexity lower bounds for the emptiness and universality problems for \qfac{s}.

The intersection non-emptiness problem for DFAs is \PSPACE-complete \cite{Kozen:77:SFCS}. Since the complement of a DFA can be obtained by swapping its accepting and rejecting states, the union universality problem for DFAs is \PSPACE-complete as well. Specifically, given DFAs $\mathcal{A}_1, \dots, \mathcal{A}_k$ over $\Sigma$, deciding whether $\bigcup_{i=1}^k L(\mathcal{A}_i) = \Sigma^*$ is \PSPACE-complete. This union is accepted by an NFA constructed by combining the individual DFAs with a newly introduced initial state. For every $(\bowtie, \lambda) \in \cmp \times [0, 1]$, we can construct a \qfac $\mathcal{A}$ that simulates this NFA such that $L_{\bowtie \lambda}(\mathcal{A}) = \bigcup_{i=1}^k L(\mathcal{A}_i)$. Accordingly, the universality problem for \qfac{s} is \PSPACE-hard even when the dimension is 2. The detailed proof is provided in \Cref{app:complexity}.
By \Cref{lem:language_symmetric}, it follows that the emptiness problem for \qfac{s} satisfies the same complexity lower bound.

\begin{theorem}
    The emptiness for \qfac{s} is \PSPACE-hard, even when the dimension is 2. 
\end{theorem}

Next, we derive the \NP complexity lower bound for the emptiness problem of flat \qfac{s}. Given a directed graph $G$, a path in $G$ is a \emph{Hamiltonian path} if it visits every node in $G$ exactly once. It is known that the \emph{Hamiltonian path problem} is \NP-complete. Suppose $G$ has $n$ nodes $\{v_1, \dots, v_n\}$. Let $S = \{\theta_1, \dots, \theta_n\}$ be a set of $n$ distinct values in $\mathbb{R}$ such that the sum of $n$ values from $S$ equals $\frac{\pi}{2}$ if and only if each value in $S$ appears exactly once. Thus, if the underlying structure of a \qfac $\mathcal{A}$ is identical to $G$, and entering node $v_i$ rotates the quantum state by an angle $\theta_i$ in the $xy$-plane, then $G$ has a Hamiltonian path if and only if there exists a run of length $n$ starting from state $(1,0)^\top$ and ending in $(0,1)^\top$ (i.e., achieving a total rotation of exactly $\frac{\pi}{2}$). 
In general, such an $\mathcal{A}$ is not necessarily flat. However, this can be addressed by constructing $n$ layers $\mathcal{A}_1, \dots, \mathcal{A}_n$, intuitively, each being a copy of $G$, and enforcing that transitions from $\mathcal{A}_i$ only lead to $\mathcal{A}_{i+1}$ for each $i < n$. This layered structure ensures flatness. Accordingly, we can derive the \NP-hardness of the emptiness problem for flat \qfac{s} on $\mathcal{H}^2$. The detailed proof is provided in \Cref{app:complexity}.

\begin{restatable}{theorem}{thmEmpNp}\label{thm:emp_NP}
    The emptiness problem for flat \qfac{s} is \NP-hard.
\end{restatable}

Note that the \qfac constructed in the above theorem is one-loop and its set of matrices generate a finite group. Therefore, the complexity in \cref{thm:emp_NP} extends to one-loop \qfac and flat and finite-group \qfac. 

For any one-loop \qfac ${\cal A} = (\Sigma,Q,q_{0},\ket{\psi_{0}},(U_{a}),(P_{q}))$, we can construct $\abs{Q}$ many unary \qfa $({\cal A}_{q})_{q\in Q}$ such that $L({\cal A}) = \bigcup_{q\in Q}L({\cal A}_{q})$. For each $q\in Q$, ${\cal A}_{q}$ captures the word that ends at state $q$. If $q$ is before or on the looping state then the construction of ${\cal A}_{q}$ is trivial. For state $q$ that is after the looping state, the initial quantum state is fixed to be a unique state just before entering the loop in ${\cal A}$ and unitary matrix is chosen to be a looping unitary matrix of ${\cal A}$. 
For a projection $P$, let $w$ be the word such that ${\cal A}$ loops on the looping state while reading $w$ and then moves to $q$, and define $P' = U_{w}^{\dagger}P U_{w}$. Since $U_w^\dagger$ is unitary, it preserves the norm, and hence $\norm{P U_{w} \ket{\psi}} =\norm{U^{\dagger}_{w}PU_{w}\ket{\psi}}$. Therefore, the accepting probability is preserved for any word that reaches $q\in Q$. This yields the following proposition.

\begin{proposition}\label{prop:oneLoop_unary}
    For any one-loop \qfac ${\cal A}=(\Sigma,Q,q_{0},\ket{\psi_{0}},(U_{a}),(P_{q}))$, it can be decomposed into $\abs{Q}$ many unary \qfa ${\cal A}_{q}$ such that $L({\cal A}) = L({\cal A}_{q})$.
\end{proposition}

As an immediate consequence, we have the following corollary.

\begin{corollary}
    The emptiness problem for unary \qfa{s} is \NP-hard.
\end{corollary}

The results of this section are summarized in \Cref{tb:summation}, excluding the trivial $(\bowtie,\lambda)$.

\section{An NIP-Based Algorithm for One-Loop \qfac}\label{sec:algo}
In this section, we present an approach for checking language emptiness of one-loop \qfac based on \emph{nonlinear integer programming (NIP)}. We present a systematic way of constructing potential witnesses based on the Taylor approximation of the cosine function, and we show that the witness is always sound and \emph{uniform semi-complete} for strict inequalities. The search is based on solving a series of NIP problems with incremental degree of approximation. We implemented a bounded searching algorithm and tested it over various examples, the results show that the algorithm scales to non-trivial instances.

The section is organized as follows: in~\Cref{subsec:nip_encoding}, we present a reduction from the non-emptiness problem of a one-loop \qfac into an NIP problem instance which involves a linear combination of cosine functions. Then in~\Cref{subsec:relaxation}, we propose a polynomial relaxation by exploiting properties of the cosine function. We show that the relaxation provides a sound witness for the original problem, and the witness is uniform semi-complete with strict inequalities. Finally in~\Cref{subsec:implementation}, we present our implementation details and some experimental results.

\subsection{NIP Encoding for Emptiness Problem of One-Loop \qfac}\label{subsec:nip_encoding}
We start with the reduction from the emptiness problem of a one-loop \qfac into an NIP problem instance. By~\cref{prop:oneLoop_unary}, it suffices for us to consider unary \qfa, and it will simplify the presentation as well. Let ${\cal A} = (U,P,\ket{\psi_{0}})$ be a unary \qfa, the alphabet and states are both singleton and thus omitted. The emptiness problem for ${\cal A}$ is equivalent to the satisfiability of the following formula:
\begin{equation}\label{eqt:matrix_form}
    \exists t\in\mathbb{N}, \norm{P U^{t}\ket{\psi_{0}}}\bowtie \lambda.
\end{equation}

By the Eigen-decomposition of $U$ and some simple algebraic manipulation, the term $\norm{P U^{t}\ket{\psi_{0}}}$ can be equivalently expressed using only real numbers and trigonometry functions, as stated in the following lemma.
\begin{restatable}{lemma}{CosLinearForm}\label{prop:cos_linear_form}
Let $P$ be a projection matrix, $U$ be a unitary matrix and $\psi$ be an initial state vector. Then for any $t\in\mathbb{N}$, 
    \begin{equation}\label{eq:cos_linear_form}
        \norm{P U^{t}\ket{\psi}} = C + \sum_{1\leq i < j\leq n}a_{ij}\cos((\theta_{j}-\theta_{i})t+\phi_{ij}),
    \end{equation}
    where $C, (\theta_{i})_{i}, (\phi_{ij})_{i,j}\in\mathbb{R}$ and $a_{ij}\in\mathbb{R}_{>0}$ are constants determined by $P, U, \psi$ and can be effectively computed in polynomial time in the size of $U$.
\end{restatable}

By substituting the expression in~\Cref{prop:cos_linear_form} into~\cref{eqt:matrix_form}, the problem reduces to an existential theory over integers with trigonometry functions. We further translate it into an NIP problem as follows:
\begin{equation}\label{eq:opt-cos}
    \opt_{t\in \mathbb{N}} \quad C + \sum_{1\leq i < j\leq n}a_{ij}\cos((\theta_{j}-\theta_{i})t+\phi_{ij})
\end{equation}
where $\opt$ is $\max$ when ${\bowtie}\in\{>,\geq\}$ and $\min$ when ${\bowtie}\in\{<,\leq\}$. It is clear that the original formula in~\cref{eqt:matrix_form} is satisfiable if and only if the optimal value of the above optimization problem satisfies the corresponding inequality.


\subsection{Sound and Semi-Complete Taylor Approximation}\label{subsec:relaxation}
In practice, solving~\cref{eq:opt-cos} could be slow since it involves a linear combination of trigonometry functions, for which most of the existing solvers have only limited support. Polynomials, on the other hand, are well-studied and optimized in many solvers. Therefore, it is natural to consider a relaxation from~\cref{eq:opt-cos} to a polynomial optimization problem. Below, we present a relaxation based on the Taylor series of the the cosine function. We start with recalling the Taylor expansion of cosine function and an inequality associated with it. 
\begin{restatable}{lemma}{LowerAndUpperBound}\label{prop:cosine}
    For any $x\in\mathbb{R}$,  $\cos x = \sum_{k=0}^{\infty}\frac{(-1)^{k}}{(2k)!}x^{2k}$ and the following holds:
    \begin{equation}\label{eq:cosine_inequality}
        \sum_{k=0}^{m-1}\bigg(\frac{x^{4k}}{(4k)!}-\frac{x^{4k+2}}{(4k+2)!}\bigg)\leq \cos x \leq \sum_{k=0}^{m-1}\bigg(\frac{x^{4k}}{(4k)!}-\frac{x^{4k+2}}{(4k+2)!}\bigg) + \frac{x^{4m}}{(4m)!}.
    \end{equation}
    We denote by $L_{m}(x)\coloneqq \sum_{k=0}^{m-1}\big(\frac{x^{4k}}{(4k)!}-\frac{x^{4k+2}}{(4k+2)!}\big)$ and $U_{m}(x)\coloneqq L_{m}(x) + \frac{x^{4m}}{(4m)!}$ as the ($m$-th order) lower bound and upper bound polynomial, respectively.
\end{restatable}

By directly substituting each term $\cos((\theta_{j}-\theta_{i})t+\phi_{ij})$ in~\cref{eq:opt-cos} with $L_{m}((\theta_{j}-\theta_{i})t+\phi_{ij})$ or $U_{m}((\theta_{j}-\theta_{i})t+\phi_{ij})$, we obtain the following polynomial integer optimization problem:
\begin{equation}\label{eq:opt-relaxed}
    \opt_{t\in \mathbb{N}} \quad C + \sum_{1\leq i < j\leq n}a_{ij}F_{m}((\theta_{j}-\theta_{i})t+\phi_{ij})
\end{equation}
where $F_{m}$ is $L_{m}$ when ${\bowtie}\in\{>,\geq\}$ and $U_{m}$ when ${\bowtie}\in\{<,\leq\}$. By~\Cref{prop:cos_linear_form}, we have that $a_{ij}\geq 0$ for all $i,j$, and so it is not hard to see that the substitution preserves the ordering. Therefore, if the optimal value in~\cref{eq:opt-relaxed} is ${\bowtie}\lambda$, then it forms a sound witness for the original problem, that is there exists some $t\in\mathbb{N}$ such that $\norm{P U^{t}\ket{\psi}}\bowtie \lambda$. Hence, we have the following soundness result.

\begin{restatable}
{theorem}{soundness}
    Let $C, a_{ij}, \theta_{i}, \phi_{ij}$ be the constants defined in~\Cref{prop:cos_linear_form}. For any $m\in\mathbb{N}$, we have
    $$\max_{t\in\mathbb{N}}\quad C + \sum_{1\leq i < j\leq n}a_{ij}L_{m}((\theta_{j}-\theta_{i})t+\phi_{ij})\bowtie\lambda\implies\exists t\in\mathbb{N},\,\norm{P U^{t}\ket{\psi}}\bowtie \lambda,$$
    where ${\bowtie}\in\{>,\geq\}$, the same holds for ${\bowtie}\in\{<,\leq\}$ mutatis mutandis.
\end{restatable}

A useful property of Taylor expansion is that it can be made as tight as possible by increasing the order of approximation $m$. As a direct result of Lagrange's remainder theorem, the witnesses of the form in~\cref{eq:opt-relaxed} is also \emph{semi-complete} for strict inequality. That is, if for any $t\in\mathbb{N}$, $\norm{P U^{t}\ket{\psi}} \bowtie \lambda$ for some ${\bowtie}\in\{>,<\}$, then there exists some $m\in\mathbb{N}$ such that the optimal value of~\cref{eq:opt-relaxed} is $\bowtie\lambda$. We have the following completeness result.

\begin{restatable}
{theorem}{completeness}
    Let $C, a_{ij}, \theta_{i}, \phi_{ij}$ be the constants defined in~\Cref{prop:cos_linear_form}. For any $t\in\mathbb{N}$, and $\varepsilon > 0$, there exists some $m\in\mathbb{N}$ such that 
    $$\norm{P U^{t}\ket{\psi}} >\lambda + \varepsilon\implies C + \sum_{1\leq i < j\leq n}a_{ij}L_{m}((\theta_{j}-\theta_{i})t+\phi_{ij}) > \lambda.$$
    The same holds for $<$ mutatis mutandis.
\end{restatable}

Note that the above semi-completeness result is not uniform. That means when the solution $t$ is large, a higher degree polynomial $F_{m}$ is required to ensure~\cref{eq:opt-relaxed} forms a witness. Now, we show that such a completeness result can be made uniform via adding auxiliary variables in the optimization problem. The key idea is to exploit the periodicity of the cosine function. Instead of~\cref{eq:opt-relaxed}, we consider the following optimization problem:
\begin{equation}\label{eq:opt-relaxed-periodic}
    \begin{aligned}
        \opt_{t\in \mathbb{N}, (K_{ij})_{i,j}\in\mathbb{Z}}&\quad C + \sum_{1\leq i < j\leq n}a_{ij}F_{m}((\theta_{j}-\theta_{i})t+\phi_{ij}+2K_{ij}\pi)\\
        \text{subject to}&\quad -\pi < (\theta_{j}-\theta_{i})t+\phi_{ij}+2K_{ij}\pi\leq \pi\quad\forall i,j.
    \end{aligned}
\end{equation}
Let $\tau_{ij}(t) = (\theta_{j}-\theta_{i})t+\phi_{ij}$.
Since the cosine function is $2\pi$-periodic, we have $L_{m}(\tau_{ij}(t)+2K_{ij}\pi)\leq\cos(\tau_{ij}(t)+2K_{ij}\pi) = \cos(\tau_{ij}(t))\leq U_{m}(\tau_{ij}(t)+2K_{ij}\pi)$ for any $K_{ij}\in\mathbb{Z}$. Therefore, we shall see that the optimal value of~\cref{eq:opt-relaxed-periodic} still forms a sound witness for the non-emptiness problem. By restricting the range of the term $\tau_{ij}(t)+2K_{ij}\pi$, the error of approximation between $F_{m}$ and $\cos$ can be bounded by a constant independent of $t$. This gives the following result.

\begin{restatable}
{theorem}{uniformcompleteness}
    Let $C, a_{ij}, \theta_{i}, \phi_{ij}$ be the constants defined in~\Cref{prop:cos_linear_form}. For any $\varepsilon >0$, there exists $m\in\mathbb{N}$, such that for any $t\in\mathbb{N}$, there exists a unique set of $(K_{ij})_{i,j}\in\mathbb{Z}$ such that for all $i<j$,$(\theta_{j}-\theta_{i})t+\phi_{ij}+2K_{ij}\pi\in (-\pi,\pi]$,
    and the following holds:
    $$\norm{P U^{t}\ket{\psi}} >\lambda + \varepsilon\implies C + \sum_{1\leq i < j\leq n}a_{ij}L_{m}((\theta_{j}-\theta_{i})t+\phi_{ij}+2K_{ij}\pi) > \lambda.$$
    The same holds for $<$ mutatis mutandis.
\end{restatable}

\subsection{Implementation}\label{subsec:implementation}
We implemented a witness searching algorithm based on~\cref{eq:opt-relaxed-periodic} with bounded degree using Python 3.10. In particular, we rely on the MindtPy solvers~\cite{mindtpy} implemented in the Pyomo package~\cite{pyomo} for solving the NIP problem. The solver utilizes the global outer approximation (GOA) strategy in~\cite{Kesavan2004-kv} to decompose the NIP problem, and GLPK~\cite{glpk} and IPOPT~\cite{Wachter2006-bc} as the backend solvers to solve the subproblems. A pseudo-code is provided in \Cref{sec:app_experiments} and the source code is available \href{https://github.com/yctsai2727/NIP-based-Witness-Searching-for-One-loop-QFAC}{here}. All experiments were run on a high-performance cluster (AMD EPYC 7702P processor with 128 GB of RAM).

\paragraph{Model.} We tested our implementation on three models, from simple to complicated: \texttt{Simple-Rotation}, \texttt{Grover} and \texttt{Quantum-Walk}, with varying parameters. All three models contain a loop of quantum operations, with a measurement performed at the end of the loop. The model \texttt{Simple-Rotation} is a toy model that involves a rotation matrix with a fixed degree on a 2D plane. The models \texttt{Grover} and \texttt{Quantum-Walk} are more involved that capture the behavior of Grover's algorithm and quantum random walk~\cite{Dai:24:CAV}, respectively. Both \texttt{Grover} and \texttt{Quantum-Walk} have two parameters: the number of qubits $k$ and the reachable state to be tested. More details are provided in the \Cref{sec:app_experiments}.

\begin{table}[tbp]\centering
    \scalebox{.8}{
        \begin{tabular}{ccccccc}\toprule
	 Instance&\# qubits&Preprocess(s)&Solving time(s)&Total time(s)&deg. of approx.&Trials solved\\\midrule
	Simple-Rotation-1 & $1$ & $<0.1$ & $<0.1$ & $<0.1$ & $1$ & 1/1 \\
	Simple-Rotation-2 & $1$ & $<0.1$ & $<0.1$ & $<0.1$ & $1$ & 1/1 \\
	Simple-Rotation-3 & $1$ & $<0.1$ & $<0.1$ & $<0.1$ & $1$ & 1/1 \\\midrule
	Grover-2 & $2$ & $<0.1$ & $<0.1$ & $<0.1$ & $1$ & 10/10 \\
	Grover-3 & $3$ & $<0.1$ & $0.77$ & $0.77$ & $1$ & 10/10 \\
	Grover-4 & $4$ & $<0.1$ & $60.83$ & $60.83$ & $2$ & 10/10 \\
	Grover-5 & $5$ & $<0.1$ & $0.29$ & $0.30$ & $1$ & 10/10 \\
	Grover-6 & $6$ & $<0.1$ & $0.22$ & $0.23$ & $1$ & 10/10 \\
	Grover-7 & $7$ & $<0.1$ & $0.34$ & $0.39$ & $1$ & 10/10 \\
	Grover-8 & $8$ & $0.19$ & $0.30$ & $0.49$ & $1$ & 10/10 \\
	Grover-9 & $9$ & $0.81$ & $0.58$ & $1.39$ & $1$ & 10/10 \\
	Grover-10 & $10$ & $3.57$ & $0.44$ & $4.02$ & $1$ & 10/10 \\
	Grover-11 & $11$ & $19.11$ & $0.34$ & $19.45$ & $1$ & 10/10 \\
	Grover-12 & $12$ & $112.52$ & $0.48$ & $113.00$ & $1$ & 10/10 \\
	Grover-13 & $13$ & $760.84$ & $6.64$ & $767.48$ & $2$ & 10/10 \\\midrule
	Quantum-Walk-2 & $3$ & $<0.1$ & $7.25$ & $7.25$ & $1$ & 10/10 \\
	Quantum-Walk-3 & $4$ & $<0.1$ & $91.13$ & $91.15$ & $3$ & 5/10 \\
	Quantum-Walk-4 & $5$ & $<0.1$ & TO & TO & N/A & 0/10 \\
	\bottomrule\\
\end{tabular}
    }
    \caption{Experimental results. The columns (left to right) correspond to the number of qubits, the average time (in seconds) for preparing the NIP encoding, solving the NIP problems and their total, the largest degree of approximation $m$ used in the 10 trials and the number of successful trials, respectively.}\label{tab:exper_result}
    \vspace{-3em}
\end{table}

\paragraph{Results.} The experimental results are summarized in~\Cref{tab:exper_result}. We imposed a maximum degree of approximation, i.e., the parameter $m$, of $5$ and a time limit of 1 minute for solving the NIP problem of fixed degree. To speed up the solving process, we also implemented a simple heuristic by setting some of the auxiliary variables $K_{ij}$ to $0$ when the solvers fail to find the solution in the first place. We separate the time measurement into two parts, the preprocessing time covers the time used for computing coefficients and setting up the NIP problem instances, which covers lines 1 to 7 in~\Cref{alg:nip}, and the solving time covers the time used by the NIP solver, which covers the loop body from lines 8 to 10 in~\Cref{alg:nip}. The total time is the sum of the two parts.

Our algorithm managed to solve most of the instances within the time limit, with most of them solved within 10 seconds. For instances with many qubits, e.g. Grover-12 and Grover-13, most of the time is spent in preprocessing, which involves computing the eigen-decomposition of a $2^{k}\times 2^{k}$ matrix for a $k$ qubits instance. The time required by the NIP-solver is relatively small in most cases, which demonstrates the efficiency of the solvers for polynomial optimization problems. Surprisingly, most problems require only a small degree of approximation to achieve a sound witness.

Another observation is that our algorithm managed to solve the Grover problem up to a very large number of qubits while it failed to solve the Quantum-walk problem with $5$ qubits. We believe this is due to the fact that the underlying unitary matrix for Grover's algorithm has a large number of repeated eigenvalues, which leads to a simpler expression in~\Cref{prop:cos_linear_form} and thus easier to solve despite the large number of qubits.

\paragraph{Comparison to related works.} To better demonstrate the performance of our algorithm, we compare our results with those of~\cite{Lewis:24:arxiv,Hu:25}, which consider the dual setting of our problem---safety verification of quantum loops. In their work, they adopt the conventional approach of synthesizing barrier certificates to establish a safety guarantee. While their approach can handle circuits involving simple quantum gates (such as $Z$ and SWAP gates) with up to 6 qubits, it struggles to scale to Grover's algorithm even for 2 qubits without manually rewriting the circuit representation. In contrast, our experimental results show that our approach handles complex instances, such as Grover's algorithm up to 13 qubits, without requiring any modification to the problem representation.

\section{Conclusion}
We studied \qfac{s} and their emptiness problem under cut-point semantics.
We showed that \qfa{s} and finite automata are incomparable in expressiveness, and hence both are strictly less expressive than \qfac{s}. Furthermore, we established a pumping lemma for \qfac{s} and used it to distinguish between the class of languages recognized by \qfac{s} under strict thresholds and those recognized under non-strict thresholds.
The emptiness problem inherits the same decidability dichotomy as for \qfa{s}: decidable for strict thresholds, undecidable otherwise.
We identified decidable fragments and established complexity lower bounds.
For one-loop \qfac{s}, we gave a sound NIP-based witness procedure, that is also uniformly semi-complete for strict thresholds.



\newpage
\bibliographystyle{abbrvurl}
\bibliography{references}

\newpage

\renewcommand{\theHsection}{A\arabic{section}}
\appendix


\section{Complete \qfac{s}}\label{app:complete}

\begin{lemma}
    Let $\mathcal{A} = (\Sigma, Q, q_0, \ket{\psi_{0}}, (U_{a})_{a\in\Sigma}, (P_q)_{q\in Q}, \Delta)$ be a \qfac. For every pair $(\bowtie,\lambda)\in\cmp\times [0,1]$, there exists a complete \qfac $\mathcal{B}$ such that $L_{\bowtie \lambda}(\mathcal{B})=L_{\bowtie \lambda}(\mathcal{A})$.
\end{lemma}
\begin{proof}
    Recall that if $L_{\bowtie \lambda}(\mathcal{A}) = \emptyset$ if $(\bowtie, \lambda)$ is
    $(>,1)$ or $(<,0)$;
    and $L_{\bowtie \lambda}(\mathcal{A}) = \Sigma^*$ 
    if $(\bowtie, \lambda)$ is
    $(\ge,0)$ or $(\le, 1)$;
    For these cases, the claim holds trivially by setting
    $\mathcal{B}$ as the complete DFA for the full, or empty, language.
    
    Consider the case $\bowtie~\in \{>, \geq\}$. As with NFAs, we can make $\mathcal{A}$ complete by adding a sink state $\perp$ and setting its projection matrix $P_{\perp}$ to the null matrix. 
    Under this construction, for every word $w$ and run $\pi$, if $\pi$ ends in $\perp$, then $\val(w, \pi, \ket{\psi_0}) = 0$. Consequently, these runs do not contribute to the acceptance probability and the language remains unchanged.

    Next, consider the case where $\bowtie~\in \{<, \leq\}$. Based on the definition of \qfac{s}, we must account for all possible runs on $w$. In particular, we distinguish between the following two situations for a given $w$:
    \begin{enumerate}[(a)]
        \item $\mathcal{A}$ has no valid runs on $w$.
        \item $\mathcal{A}$ has at least one valid run on $w$, though it may also have runs on prefixes of $w$ that cannot be extended to valid runs on the complete word.
    \end{enumerate}
    We aim to construct a \qfac $\mathcal{B}$ such that for all non-trivial $(\bowtie, \lambda)$ with $\bowtie~\in \{<, \leq\}$, the following hold: 
    For case (a), $\val_{\mathcal{B}}(w)=1$ and $w \notin L_{\bowtie \lambda}(\mathcal{B})$; for case (b), $\val_{\mathcal{B}}(w) = \val_{\mathcal{A}}(w)$. One can immediately see that a simple sink state construction is insufficient: naively setting $P_{\perp}$ as the null matrix results in the inclusion of words from case (a) in the language, whereas setting it as the identity matrix may incorrectly exclude words from case (b). 
    
    To remedy these problems, we shall construct a \qfac{} that keeps track of and distinguishes between cases (a) and (b). Specifically, whether the word $w$ still possesses a valid run in ${\cal A}$. A simple observation is that tracking this property is achievable by the standard subset construction, namely by checking whether it reaches the empty set or not. Therefore, by combing the sink state construction with the subset construction, we can distinguish between cases (a) and (b) and apply the corresponding projection in each case. Formally, the product of these two constructions is given as follows:
    
    Let $\Delta_{(q,a)}\coloneqq \{q'\in Q\mid (q,a,q')\in\Delta\}$ and $\Delta_{(S,a)} = \bigcup_{q\in S}\Delta_{(q,a)}$. We construct a \qfac ${\cal B} = (\Sigma,Q'\subseteq 2^Q\times (Q\cup\{\perp\}),(\{q_{0}\},q_{0}),\ket{\psi_{0}},(U_{a})_{a\in\Sigma},(P_{q})_{q\in Q'},\Delta')$ such that
    \begin{itemize}
        \item $Q'\coloneqq \{(S, q)\in 2^{Q}\times (Q\cup\{\perp\})\mid (S\subseteq Q\land q\in S)\lor q = \perp\}$
        \item For each non-empty set $S$, state $q\in S$, and alphabet $a\in\Sigma$:
        \begin{itemize}
            \item if $\Delta_{(q,a)}\neq\emptyset$, then for any $q'\in \Delta_{(q,a)}$, $((S,q),a, (\Delta_{(S,a)},q'))\in\Delta'$,
            \item otherwise, $((S,q),a,(\Delta_{(S,a)},\perp)\in\Delta'$.
        \end{itemize}
        \item For each non-empty set $S$ and alphabet $a\in\Sigma$, $((S,\perp),a,(\Delta_{(S,a)},\perp)\in\Delta'$.
        \item For each alphabet $a\in\Sigma$, $((\emptyset,\perp),a,(\emptyset,\perp))\in\Delta'$.
        \item For each state $q\in Q$, $P_{(S,q)} = P_{q}$.
        \item For each subset $S$, if $S\neq\emptyset$, then $P_{(S,\perp)} = O$, otherwise, $P_{(\emptyset,\perp)} = I$, where $O$ and $I$ are the null and the identity matrices, respectively.
    \end{itemize}
    By construction, for every $p' \in Q'$ and $a \in \Sigma$, there exists a $q' \in Q'$ such that $(p', a, q') \in \Delta'$, ensuring that $\mathcal{B}$ is complete.

    To show $L_{\bowtie\lambda}(\mathcal{B}) = L_{\bowtie\lambda}(\mathcal{A})$, we begin by proving $L_{\bowtie\lambda}(\mathcal{B}) \subseteq L_{\bowtie\lambda}(\mathcal{A})$. Let $w \in L_{\bowtie \lambda}(\mathcal{B})$. By construction, any valid run $\pi$ of $w$ in $\mathcal{A}$ corresponds to a unique run $\pi'$ in $\mathcal{B}$ such that $\val_{\mathcal{A}}(w, \pi, \ket{\psi_0}) = \val_{\mathcal{B}}(w, \pi', \ket{\psi_0})$. Specifically, if $w = a_1 \dots a_n$ and $\pi = q_0 \dots q_n$ is a run over $Q$, then $\pi' = q'_0 \dots q'_n$, where $q'_i = (S_i,q_i)$ and $S_i$ is the set of states reachable from $q_0$ via $a_1 \dots a_i$ in $\mathcal{A}$. Since the state $(\emptyset,\bot)$ in $\mathcal{B}$ is reached from the initial state on $w$ if and only if $\mathcal{A}$ has no valid runs on $w$, it follows that for any $w$ with at least one valid run, $\val_{\mathcal{A}}(w) = \val_{\mathcal{B}}(w)$. For words with no valid runs, $\val_{\mathcal{B}}(w) = 1$ (due to the identity projection matrix), and since $\lambda < 1$ for non-trivial cases, these words are excluded from $L_{\bowtie \lambda}(\mathcal{B})$ when $\bowtie \in \{<, \leq\}$. Consequently, $\val_{\mathcal{B}}(w) \bowtie \lambda$ implies $\val_{\mathcal{A}}(w) \bowtie \lambda$, which yields $L_{\bowtie\lambda}(\mathcal{B}) \subseteq L_{\bowtie\lambda}(\mathcal{A})$.

    For the other direction, let $w\in L_{\bowtie\lambda}({\cal A})$. By definition, there must exist at least one valid run $\pi$ in ${\cal A}$ such that $\val_{\cal A}(w,\pi)\bowtie\lambda$. WLOG, let $\pi$ be the optimal run, i.e. $\val_{\cal A}(w,\pi,\ket{\psi_0}) = \val_{\cal A}(w)$. Let $\pi'$ be the unique run in ${\cal B}$ induced by $\pi$, then we have $\val_{\cal B}(w,\pi',\ket{\psi_0}) = \val_{\cal A}(w)\bowtie\lambda$. Therefore, It remains for us to show that for any other run $\pi''\neq \pi'$ induced by $w$, we also have $\val_{\cal B}(w,\pi'',\ket{\psi_0}) \bowtie \lambda$. 
    
    Assume, for the sake of contradiction, there exists a run $\pi''\neq \pi'$ of $w$ in ${\cal B}$ such that $\val_{\cal B}(w,\pi'',\ket{\psi_0})\Bar{\bowtie} \lambda$. Let $(S,q)$ be the last state in $\pi''$. We consider three cases: (1) $S$ is non-empty and $q\in S$, (2) $S$ is non-empty and $q = \perp$, and (3) $S = \emptyset$ and $q = \perp$. Clearly, case (2) and (3) are impossible: in case (2), since $P_{(S,\perp)} = O$, it follows that $\val_{\cal B}(w,\pi'',\ket{\psi_0}) = 0$ contradicting the assumption that $\val_{\cal B}(w,\pi'',\ket{\psi_0})\Bar{\bowtie} \lambda$ (recall that we exclude the trivial case $(<,0)$ so $\geq 0$ cannot occur here). Case (3) contradicts the fact that $\pi''$ is a run induced by $w$, as the left component in the state of ${\cal B}$ must correspond to the state in the subset construction. Therefore, $S$ must be non-empty and $q\in S$. By construction, $\pi''$ corresponds to a run $\pi_{\ast}$ such that $\val_{\cal B}(w,\pi'',\ket{\psi_0}) = \val_{\cal A}(w,\pi_{\ast},\ket{\psi_0})\Bar{\bowtie}\lambda$. However, our assumption that $\pi$ is optimal implies $\val_{\mathcal{A}}(w) = \max_{\pi \in \text{Path}_{\mathcal{A}}(w)} \val_{\mathcal{A}}(w, \pi, \ket{\psi_0}) \bowtie \lambda$, which contradicts the fact that $\val_{\mathcal{A}}(w, \pi_{\ast}, \ket{\psi_0}) \bar{\bowtie} \lambda$. Therefore, for any run $\pi'' \neq \pi'$, we have $\val_{\mathcal{B}}(w, \pi'') \bowtie \lambda$, and thus $\val_{\mathcal{B}}(w) \bowtie \lambda$. This implies $L_{\bowtie\lambda}(\mathcal{A}) \subseteq L_{\bowtie\lambda}(\mathcal{B})$, which concludes the proof.
    \qed
\end{proof}

\section{Missing Proofs of \Cref{sec:expressiveness}}\label{app:expressiveness}

\begin{lemma}\label{lem:dqfa_sym}
 Let $\lambda \in [0,1]$ and $\bowtie \in \cmp$. For every deterministic \qfac $\mathcal{A}$ over $\Sigma$, one can effectively construct a \qfac $\mathcal{B}$ such that $L_{\bowtie\lambda}(\mathcal{A}) = L_{\Tilde{\bowtie}1-\lambda}(\mathcal{B})$.
\end{lemma}

\begin{proof}
Since ${\cal A}$ is deterministic, for any word $w\in L_{\bowtie}({\cal A})$, $w$ induces a unique run on ${\cal A}$ and thus there exists a unique state $q\in Q$ such that $\val_{\cal A}(w) = \norm{P_{q} U_{w}\ket{\psi_{0}}}\bowtie\lambda$. By property of projection operator, we have $$\norm{P_{q} U_{w}\ket{\psi_{0}}} = 1-\norm{(I-P_{q})U_{w}\ket{\psi_{0}}}.$$
Therefore, let ${\cal A}'$ be a deterministic \qfac such that ${\cal A}'$ is the same as ${\cal A}$ except that the projection operator at state $q$ is replaced by $I-P_{q}$. It is obvious that ${\cal A}'$ must be deterministic as well. Then it is straightforward to see that $\val_{\cal A}(w)\bowtie\lambda \iff 1-\val_{\cal A}(w)\Tilde{\bowtie}1-\lambda\iff \val_{{\cal A}'}(w)\Tilde{\bowtie}1-\lambda$, thus $w\in L_{\Tilde{\bowtie} 1-\lambda}({\cal A}')$. It is also clear that ${\cal A}'$ can be constructed in polynomial time of the size of ${\cal A}$ and dimension of the underlying Hilbert space.
\qed
\end{proof}

\lemQFAnonRegular*
\begin{proof}
    We shall show that for any non-trivial pair of $(\bowtie,\lambda)\in \cmp\times [0,1]$, there exists a \qfa $\mathcal{A}$ such that $L_{\bowtie \lambda}(\mathcal{A})$ is not regular. Since the class of regular language is closed under complements, by~\Cref{lem:language_symmetric}, proving the case of $\bowtie\in\{>,<\}$ will automatically implies the case of $\bowtie\in\{\geq,\leq\}$. Furthermore, by~\Cref{lem:dqfa_sym}, it suffices for us to prove that there is a deterministic \qfa ${\cal A}$ recognizes a non-regular language for the case of $(>,\lambda)$ for $\lambda\in [0,1)$. Note that the pair $(>,1)$ is trivial and thus being excluded.
    
    In \cite{Brodsky:02:SIAM}, it is shown that there exists a 2-dimensional \qfa $\mathcal{A}$ with strict threshold $\lambda = 0$ recognizing the non-regular language $L_{\neq} = \{ w \in \Sigma^* \mid |w|_a \neq |w|_b \}$, that is, $L_{>0}(\mathcal{A}) = L_{\neq}$. We can generalize this result to any threshold $\lambda\in (0,1)$ by constructing a 3-dimensional \qfa $\mathcal{A}'$ as below.
    
    Let $\mathcal{A} = (\Sigma, Q, q_0, \ket{\psi_{0}}, (U_{a})_{a\in\Sigma}, (P_q)_{q\in Q}, \Delta)$ be the 2-\qfa recognizing $L_{\neq}$, we construct a 3-\qfa $\mathcal{B} = (\Sigma, Q, q_0, \psi'_{0}, (U'_{a})_{a\in\Sigma}, (P'_q)_{q\in Q}, \Delta)$ where 
    \[\ket{\psi'_{0}} = \begin{pmatrix}
        \sqrt{1-\lambda}\,\ket{\psi_{0}}\\\sqrt{\lambda} 
    \end{pmatrix},\]
    \[U_{a}' = \begin{pmatrix}
        U_{a}&0\\
        0&1
    \end{pmatrix},\]
    and 
    \[P_{q}'=\begin{pmatrix}
        P_{q} & 0 \\
        0 & 1
    \end{pmatrix}.\] 
    By construction, it is not hard to show that for every word $w$ and run $\pi$ starting from the initial state $q_0$ on $w$,
    \[
    \val_{\mathcal{B}}(w,\pi,\ket{\psi_0})=(1-\lambda)\cdot\val_{\mathcal{A}}(w,\pi,\ket{\psi'_0}) + \lambda.
    \]
    Since $\lambda\in (0,1)$, it follows that $\val_{\cal B}(w)>\lambda\iff \val_{\cal A}(w)>0$. Hence, $L_{>\lambda}(\mathcal{B}) = L_{\neq}$ for any $\lambda\in (0,1)$.
    \qed
\end{proof}

\lemNonQFArecognizable*
\begin{proof}
    By~\Cref{lem:language_symmetric}, it is equivalent to show that for $\bowtie\in\{>,<\}$ and any $\lambda\in [0,1]$, there is no \qfa ${\cal A}$ that accepts $L_{\neg (bb)}$ or its complement. We first show it for the case of $>$.

    For any $\lambda\in [0,1)$, assume there exists a \qfa ${\cal A}$ such that $L_{>\lambda}({\cal A})=L_{\neg (bb)}$. By definition we must have $\val_{\cal A}(aa)>\lambda$. By~\Cref{lem:qfa_pumping} with $u=v=a$, $w=bb$ and $\varepsilon = \frac{\val_{\cal A}(aa)-\lambda}{2}$, there exists a positive integer $k\in\mathbb{N}_{>0}$ such that $\abs{\val_{\cal A}(a(bb)^{k}a)-\val_{\cal A}(aa)}\leq \varepsilon$. Then it follows that
    $$\val_{\cal A}(a(bb)^{k}a)\geq \val_{\cal A}(aa)-\frac{\val_{\cal A}(aa)-\lambda}{2}>\lambda,$$
    thus leading to a contradiction. By replacing the word $aa$ as $bb$, and picking $u=v=b$ and $w=a$, one can similarly prove that the complement of $L_{\neg (bb)}$ is not recognizable as well. 

    The case of $<$ can be proven similarly by defining $\varepsilon' = -\varepsilon$.
    \qed
\end{proof}

\lemPumpingQFAC*

\begin{proof}
    Let ${\cal B}$ be the deterministic finite-state machine obtained by performing subset construction on the classical part of $\mathcal{A}$. For any word $w = a_{1}\dots a_{n}\in L_{\bowtie\lambda}({\cal A})$, it must induces a unique run $S_{0}a_{1}S_{1}\dots a_{n}S_{n}$ on ${\cal B}$ such that $S_{0} = \{q_{0}\}$ and $\val_{\cal A}(w) = \max_{q\in S_{n}}\norm{P_{q}U_{w}\ket{\psi_{0}}}$. Furthermore, by the assumption that any \qfac{s} are complete, we also have $S_{i}\neq\emptyset$ for every $i$. Now suppose $n = 2^{\abs{Q}}$, by pigeonhole principle, there must exist two distinct indices $i,j\in [0,2^{\abs{Q}}-1]$ such that $S_{i} = S_{j}$. Let $w_{1}= a_{1}\dots a_{i}$, $w_{2} = a_{i+1}\dots a_{j}$, and $w_{3} = a_{j+1}\dots a_{n}$. It is clear that we have $\abs{w_{1}w_{2}}\leq j<n$, and $\abs{w_{2}} = j-i > 0$. Now, it suffices for us to verify that for every $\varepsilon > 0$, there exists $k\in\mathbb{N}_{>0}$ such that $\abs{\val_{\cal A}(w_{1}(w_{2})^{k}w_{3})-\val_{\cal A}(w)}\leq \varepsilon$.

    It can be shown using a similar argument as in~\cite{Blondel:05:SIAM}. Specifically, we know that the set $\{U_{w_{2}}^{k} \mid k\in\mathbb{N}\}$ is compact and thus it admits a Cauchy subsequence. Hence, there must exist some distinct $k_{1}<k_{2}\in\mathbb{N}_{>0}$ such that $\norm{U_{w_{2}}^{k_{2}}-U_{w_{1}}^{k_{1}}}\leq \frac{\varepsilon}{2}$, where $\norm{\cdot}$ is the operator norm with respect to the Euclidean norm. By property of unitary matrix, we have $\norm{U_{w_{2}}^{k_{2}-k_{1}+1}-U_{w_{2}}}\leq \frac{\varepsilon}{2}$. Let $k = k_{2}-k_{1}+1$, we have
    \begin{align*}
        &\abs{\val_{\cal A}(w_{1})(w_{2})^{k}w_{3}-\val_{\cal A}(w)}\\
        & = \abs{\max_{q\in S_{n}}\norm{P_{q}U_{w_{1}}U_{w_{2}}^{k}U_{w_{3}}\ket{\psi_{0}}}-\max_{q\in S_{n}}\norm{P_{q}U_{w}\ket{\psi_{0}}}}\\
        &=\abs{(\max_{q\in S_{n}}\abs{\abs{P_{q}U_{w_{1}}U_{w_{2}}^{k}U_{w_{3}}\ket{\psi_{0}}}})^{2}-(\max_{q\in S_{n}}\abs{\abs{P_{q}U_{w}\ket{\psi_{0}}}})^{2}}\\
        &=\abs{\max_{q\in S_{n}}\abs{\abs{P_{q}U_{w_{1}}U_{w_{2}}^{k}U_{w_{3}}\ket{\psi_{0}}}}+\max_{q\in S_{n}}\abs{\abs{P_{q}U_{w}\ket{\psi_{0}}}}}\abs{\max_{q\in S_{n}}\abs{\abs{P_{q}U_{w_{1}}U_{w_{2}}^{k}U_{w_{3}}\ket{\psi_{0}}}}-\max_{q\in S_{n}}\abs{\abs{P_{q}U_{w}\ket{\psi_{0}}}}}\\
        &\leq 2\abs{\max_{q\in S_{n}}\abs{\abs{P_{q}U_{w_{1}}U_{w_{2}}^{k}U_{w_{3}}\ket{\psi_{0}}}}-\max_{q\in S_{n}}\abs{\abs{P_{q}U_{w}\ket{\psi_{0}}}}}\\
        &\leq2\abs{\max_{q\in S_{n}}\bigg(\abs{\abs{P_{q}U_{w_{1}}U_{w_{2}}^{k}U_{w_{3}}\ket{\psi_{0}}-P_{q}U_{w}\ket{\psi_{0}}}}\bigg)}\\
        &=2\max_{q\in S_{n}}\abs{\abs{P_{q}U_{w_{1}}(U_{w_{2}}^{k}-U_{w_{2}})U_{w_{3}}\ket{\psi_{0}}}}\\
        &\leq 2\max_{q\in S_{n}}\abs{\abs{P_{q}}}\cdot\abs{\abs{U_{w_{1}}}}\cdot \abs{\abs{(U_{w_{2}}^{k}-U_{w_{2}})}}\cdot \abs{\abs{U_{w_{3}}\ket{\psi_{0}}}}\leq\varepsilon.
    \end{align*}\qed
\end{proof}

\lemBideterministicQFACtoQFA*
\begin{proof}
    Let $\mathcal{A}_{\text{under}}$ be the underlying automaton of an $n$-dimensional \qfac $\mathcal{A}$ with $m$ states. Since $\mathcal{A}_{\text{under}}$ is bideterministic, for every symbol $a \in \Sigma$, the transition relation of $\mathcal{A}_{\text{under}}$ is a bijection on the states. This mapping corresponds to an $m \times m$ permutation matrix $U_a$, and every permutation matrix is unitary. Let $\mathcal{C}$ be an $m$-dimensional \qfa over $\Sigma$ where each $a \in \Sigma$ is associated with the permutation matrix $U_a$. We can use $\mathcal{C}$ to simulate the transitions of $\mathcal{A}_{\text{under}}$, where each element in the standard basis of the $m$-dimensional space corresponds to a state of $\mathcal{A}$.

    Next, we define $\mathcal{B}$. For ease of exposition, we assume that for each state $q$ of $\mathcal{A}$, the projection matrix is $P_q = \sum_{j=1}^{i}\ket{j}\bra{j}$ for some $0 \leq i \leq n$. Let $\mathcal{B}$ be an $(m+2n)$-dimensional \qfa. Intuitively, for each $a \in \Sigma$, the unitary matrix of $\mathcal{B}$ acts on two main components: (1) the first $m$ dimensions, which act as the permutation matrix $U_a$ from $\mathcal{C}$; and (2) the remaining $2n$ dimensions, which manage the quantum state of $\mathcal{A}$. Although the quantum operations of $\mathcal{A}$ require only $n$ dimensions, we introduce $n$ additional dimensions to account for the fact that the projection matrices $P_q$ may vary across different states $q$ in the \qfac $\mathcal{A}$. 
    Since \qfa $\mathcal{B}$ possesses only a single state, these $n$ extra dimensions (initially set to zero) allow us to normalize the projection. Specifically, if the current classical state is $q$ and the rank of the projection $P_q$ is $k$, then the entries from the $(m+k+1)$-th to the $(m+n)$-th positions are swapped with the $(n-k)$ zero-valued entries in the final $n$ dimensions. This is achieved using a sequence of controlled-swap operations, which can be composed from basic quantum operations. A similar process must be performed before the quantum operation corresponding to each symbol $a$. Accordingly, by setting $P = \sum_{j=m+1}^{m+n}\ket{j}\bra{j}$ as the projection matrix of $\mathcal{B}$, we ensure that $\val_\mathcal{B}(w) = \val_\mathcal{A}(w)$ for all $w \in \Sigma^*$. 
    \qed
\end{proof}

\section{Missing Proofs of \Cref{sec:decidability_complexity}}

\subsection{Missing Proofs of \Cref{sec:deci_qfac}}\label{app:deci_qfac}

\lemNFAQFAintersect*

\begin{proof}
    We construct the \qfa ${\cal A}_{q}$ and an NFA $\mathcal{B}_{q}$ as follows: let ${\cal A}_{q} = (\Sigma,\ket{\psi_{0}}, (U_{a})_{a\in\Sigma}, P_{q})$ and ${\cal B}_{q} = (\Sigma, Q, q_{0},\Delta, F = \{q\})$. Each component of ${\cal A}_{q}$ and ${\cal B}_{q}$ is a component from ${\cal A}$.

    We first show the only if direction. For any $w\in L_{>\lambda}({\cal A})[q]$, by definition, the following must holds:
    \begin{itemize}
        \item there exists a run $q_{0}a_{1}q_{1}\dots a_{n}q_{n}=q$ in ${\cal A}$, and
        \item $\norm{P_{q} U_{w}\ket{\psi_{0}}}>\lambda$.
    \end{itemize}
    By construction of ${\cal B}_{q}$, the first condition implies that $w$ is accepted by ${\cal B}_{q}$. Further, the second condition implies that $w$ is $(\bowtie,\lambda)$-accepted by ${\cal A}_{q}$. Therefore, $w\in L_{> \lambda}({\cal A}_{q})\cap L({\cal B}_{q})$.

    Now for the if direction. Suppose $w\in L_{> \lambda}({\cal A}_{q})\cap L({\cal B}_{q})$. By construction of ${\cal A}_{q}$, we have $\norm{P_{q} U_{w}\ket{\psi_{0}}}>\lambda$. On the other hand, by construction of ${\cal B}_{q}$ $w$ induces a run $q_{0}a_{1}q_{1}\dots a_{n}q_{n}$ with $q_{n} = q$ such that $(q_{k-1}, a_{k}, q_{k})\in\Delta$ for any $k=1,\dots, n$. Therefore, combining both, we have $w\in L_{>\lambda}({\cal A})[q]$.\qed
\end{proof}

\thmEmptDeci*

\begin{proof}

Let ${\cal A} = (\Sigma, Q, q_{0}, \Delta, (P_{q})_{q\in Q}, \ket{\psi_{0}})$. We first show that the decomposition in~\cref{eq:decomposition_gt} is correct. By definition it is clear that $L_{>\lambda}({\cal A})\subseteq \bigcup_{q\in Q} L_{>\lambda}({\cal A})[q]$. For the other direction, for any $w\in \bigcup_{q\in Q} L_{>\lambda}({\cal A})[q]$, then there exists some run $\pi$ such that $\pi$ ends in the state $q$ and $\val_{\cal A}(w,\pi,\ket{\psi_{0}})>\lambda$. Then we must have $\val_{\cal A}(w)=\max_{\pi\in Path_{\cal A}(w)}\val_{\cal A}(w,\pi,\ket{\psi_{0}})\geq \val_{\cal A}(w,\pi,\ket{\psi_{0}}) >\lambda$, hence $w\in L_{>\lambda}({\cal A})$. Therefore, we have 
$$L_{>\lambda}({\cal A}) = \bigcup_{q\in Q} L_{>\lambda}({\cal A})[q].$$

By~\cref{lem:NFA_QFA_Intersection}, for any $q\in Q$, there exist a \qfa ${\cal A}_{q}$ and an NFA ${\cal B}_{q}$ such that $L_{>\lambda}({\cal A})[q] = L_{> \lambda}({\cal A}_{q})\cap L({\cal B}_{q})$. Therefore, the problem can be further reduced to checking the emptiness problem of the intersection of a \qfa and an NFA. 

Any NFA can be represented as a regular language, which is linear context-free~\cite{Chomsky:58:IC}. Here we sketch a proof for self-containedness. Let ${\cal B} = (\Sigma,Q,q_{0},\Delta,F)$ be an NFA. We construct a context-free grammar ${\cal G}$ as follows: let $Q$ be the set of non-terminals, $\Sigma$ be the set of terminals, $q_{0}\in Q$ be the starting symbol and the production rules are defined as follows:
\begin{itemize}
    \item for each $(q,a,q')\in\Delta$, we define a production rule $q\to aq'$,
    \item and for each $q\in F$, we define a production rule $q\to\epsilon$.
\end{itemize}
It is straightforward to verify that ${\cal G}$ is linear and generates the language $L({\cal B})$. 

Therefore, by applying the above construction, we can obtain a context-free grammar ${\cal G}_{q}$ such that $L({\cal B}_{q}) = L({\cal G}_{q})$. By result of~\cite{Bertoni:13:DLT}, the emptiness problem of $L_{> \lambda}({\cal A}_{q})\cap L({\cal G}_{q})$ is decidable for any $q\in Q$. Finally, since $Q$ is finite, we can always check whether $L_{>\lambda}({\cal A})[q]$ is empty for each $q\in Q$ and hence decide whether $L_{>\lambda}({\cal A})$ is empty or not. \qed
\end{proof}

\IntersecMultiQFA*

\begin{proof}
    Let $(U_{a})_{a\in\Sigma}$ be the set of transition operators for ${\cal A}_{i}$, $\ket{\psi_{0}}$ be the intial state and $P_{i}$ be the projection operator for the $i$-th \qfa. Further we denote $\Gamma({\cal G})$ the set of matrices $U_{w}$ induced by a word $w\in L({\cal G})$. Let $f_{i}$ denotes the function $M\mapsto \norm{P_{i}M\ket{\psi_{o}}}$ for $i=1,\dots, K$, then the set $\bigcap_{i=1}^{K}L_{<\lambda}({\cal A}_{i})\cap L({\cal G})$ is empty if and only if the following formula is satisfiable:
    $$\forall M\in \mathbb{R}^{n\times n},\,M\in\Gamma({\cal G})\implies \bigvee_{i=1}^{K}f_{i}(M)\geq\lambda.$$
    The above formula is also equivalent to deciding if $\Gamma({\cal G})\subseteq \bigcup_{i=1}^{K} f_{i}^{-1}([\lambda,1])$ where $f^{-1}_{i}$ denotes the pre-image of $f_{i}$. Since for any $i$, $f_{i}$ is a continuous function and $[\lambda,1]$ is a closed set for any $\lambda\in [0,1]$, thus $\bigcup_{i=1}^{K}f_{i}^{-1}([\lambda,1])$ is also closed. Therefore, we have $\Gamma({\cal G})\subseteq \bigcap_{i=1}^{K} f_{i}^{-1}([\lambda,1])$ if and only if $\Bar{\Gamma({\cal G})}\subseteq \bigcap_{i=1}^{K} f_{i}^{-1}([\lambda,1])$ where $\overline{\Gamma({\cal G})}$ denotes the topological closure of $\Gamma({\cal G})$. As a result, it suffices to check the satisfiability of the following formula:
    $$\forall M\in \mathbb{R}^{n\times n},\,M\in\Bar{\Gamma({\cal G})}\implies \bigvee_{i=1}^{K}f_{i}(M)\geq\lambda.$$
    By result in~\cite{Bertoni:13:DLT}, $\Bar{\Gamma({\cal G})}$ is an effectively eventually definable semialgebraic set. Therefore, the satisfiability of the above formula is decidable using the same construction as in~\cite{Bertoni:13:DLT,Blondel:05:SIAM}.\qed
\end{proof}

\thmEmptDeciFull*

\begin{proof}
    Since every \qfa is a \qfac, the undecidable result for $\bowtie\in\{\geq,\leq\}$ is trivial following from the result of~\cite{Blondel:05:SIAM}. The result for $\bowtie = >$ is also shown in~\cref{thm:decidability_emptiness}. It remains for us to show the decidable result for $\bowtie = <$.

    We shall first show that the decomposition in~\cref{eq:decomposition_lt} is correct. For any subset $S\subseteq Q$, let ${\cal B}[S]$ denotes the DFA obtained by determinizing the classical part of ${\cal A}$ with $S$ as the unique accepting state, and let ${\cal A}_{q}$ be the same \qfa defined in~\Cref{lem:NFA_QFA_Intersection}. We shall show that 
    $$L_{<\lambda}({\cal A}) = \bigcup_{S\subseteq Q, S\neq \emptyset} \left(\bigcap_{q\in S} L_{<\lambda}({\cal A}_{q})\cap L({\cal B}[S])\right).$$
    For the $\subseteq$ direction, suppose that $w\in L_{<\lambda}({\cal A})$, then for any run $\pi\in Path_{\cal A}(w)$, we have $\val_{\cal A}(w,\pi,\ket{\psi_{0}})<\lambda$. Let $S = \{q\in Q\mid \text{there exists a run $\pi$ of $w$ that ends at $q$}\}$, then by definition of ${\cal B}[S]$, we have $w\in L({\cal B}[S])$. Further, we know that for any $q\in S$ and run $\pi$ of $w$ that ends at $q$, we have $\val_{\cal A}(w,\pi,\ket{\psi_{0}}) < \lambda$ which implies $w\in L_{<\lambda}({\cal A}_{q})$. Therefore, we have $w\in \bigcap_{q\in S} L_{<\lambda}({\cal A}_{q})\cap L({\cal B}[S])$.

    For the $\supseteq$ direction, suppose that $w = a_{1}\dots a_{n}\in \bigcap_{q\in S} L_{<\lambda}({\cal A}_{q})\cap L({\cal B}[S])$ for some $S\subseteq Q$. For any run $\pi = q_{0}a_{1}q_{1}\dots a_{n}q_{n}$ of $w$ on ${\cal A}$, by definition of ${\cal B}[S]$, we have $q_{n}\in S$. Therefore, by the assumption that $w\in L_{<\lambda}({\cal A}_{q})$ for any $q\in S$, we have $\val_{\cal A}(w,\pi,\ket{\psi_{0}}) = \norm{P_{q_{n}}U_{w}\ket{\psi_{0}}} < \lambda$. Since it holds for arbitrary run induced by $w$, we also have $\val_{\cal A}(w) = \max_{\pi\in Path_{\cal A}(w)}\val_{\cal A}(w,\pi,\ket{\psi_{0}}) < \lambda$ and hence $w\in L_{<\lambda}({\cal A})$.
    
    Therefore,~\cref{eq:decomposition_lt} holds. Now for any $S\subseteq Q$, by~\cref{thm:intersection_multi_qfa}, the emptiness problem of the language $\bigcap_{q\in S} L_{<\lambda}({\cal A}_{q})\cap L({\cal B}[S])$ is decidable. Since $Q$ is finite, we have only finitely many non-empty subset $S$ of $Q$. Therefore, by enumerating all such subset $S$ and checking the emptiness of $\bigcap_{q\in S} L_{<\lambda}({\cal A}_{q})\cap L({\cal B}[S])$ for each $S$, we can decide whether $L_{<\lambda}({\cal A})$ is empty or not. \qed
\end{proof}

\subsection{Missing Proofs of \Cref{sec:deci_flat}}\label{app:flat}

\lemFlattoChain*
\begin{proof}
    Let $S_1, \dots, S_n$ be the strongly connected components (SCCs) of the underlying graph of $\mathcal{A}$. For any two SCCs $S_i$ and $S_j$, we say $S_i$ is an \emph{ancestor} of $S_j$ if there exists a path from $S_i$ to $S_j$. We say $S_j$ is a \emph{child} of $S_i$ if $S_i$ is an ancestor of $S_j$ and every ancestor of $S_j$ (other than $S_j$ itself) is also an ancestor of $S_i$.

    Suppose $S_j$ and $S_k$ are two distinct children of $S_i$. Let $\mathcal{A}^{\textit{delete}}_{S_j}$ and $\mathcal{A}^{\textit{delete}}_{S_k}$ be the \qfac{s} obtained from $\mathcal{A}$ by removing $S_j$ and $S_k$, respectively. By this construction, it follows that $L_{\bowtie \lambda}(\mathcal{A}) = L_{\bowtie \lambda}(\mathcal{A}^{\textit{delete}}_{S_j}) \cup L_{\bowtie \lambda}(\mathcal{A}^{\textit{delete}}_{S_k})$ for every $\bowtie \in \cmp$ and $\lambda \in [0,1]$. By repeatedly applying this process, we can derive that $L_{\bowtie \lambda}(\mathcal{A}) = \bigcup_{i=1}^{m}L_{\bowtie \lambda}(\mathcal{A}_i)$ for some chain \qfac{s} $\mathcal{A}_1, \dots, \mathcal{A}_m$.
    \qed
\end{proof}
\subsection{Missing Proofs of \Cref{sec:complexity}}\label{app:complexity}

\begin{restatable}{theorem}{thmUniPspace}\label{thm:uni_pspace}
    The universality problem for \qfac{s} is \PSPACE-hard, even when the dimension is 2. 
\end{restatable}
\begin{proof}
Let ${\cal A}_1,\dots, {\cal A}_k$ be $k$ DFAs over $\Sigma$, where ${\cal A}_i = (Q_i,q^{[i]}_0,F_i,\delta_i)$ for each $i=1,\dots,k$. We define a 2-dimensional \qfac $\mathcal{A} = (\Sigma, Q, q^{[0]}_0, \psi_0, (U_{a})_{a\in\Sigma}, (P_q)_{q\in Q}, \Delta)$, where:
\begin{enumerate}[$\bullet$]
    \item $Q = \{q_0^{[0]}\} \cup \bigcup_{i=1}^{k}Q_i$.
    \item $\psi_0=(\sqrt{\frac{2}{3}},\sqrt{\frac{1}{3}})$.
    \item $U_a = I$ for all $a \in \Sigma$.
    \item $P_q$ projects a vector $(x,y)^T$ to the $x$-axis if $q \in \bigcup_{i=1}^k F_i$; otherwise, $P_q$ is the zero matrix for each $q\in Q$.
    \item For each $a\in\Sigma$, the transition relation is defined as follows:
        \begin{enumerate}[$\bullet$]
            \item $(q_0,a,q)\in\Delta$ for every $q\in\{q^{[1]}_0,\dots,q^{[k]}_0\}$.
            \item $(p,a,q)\in\Delta$ if $(p,a,q)\in \delta_i$ for some $i=1,\dots,k$.
        \end{enumerate}
\end{enumerate}
By letting $\lambda=\frac{1}{2}$, we have that the union of ${\cal A}_1,\dots, {\cal A}_k$ is universal if and only if for every word $w \in \Sigma^*$, there exists an accepting run of $\mathcal{A}$ from $(q_0,\psi_0)$ with accepting probability $\geq \lambda$.
\qed
\end{proof}

\thmEmpNp
\begin{proof}
    To derive the statement, we reduce the Hamiltonian path problem to the emptiness problem of \qfac{s}. Let $G=(V,E)$ be a directed graph, where $V={v_0,\dots,v_m}$ for some $m\geq 1$. We construct a 2-dimensional \qfac $\mathcal{A}$ and choose $0<\lambda<1$ such that there exists a Hamiltonian path from $v_0$ to $v_m$ in $G$ if and only if there exists a word $w$ with $\val_{\mathcal{A}}(w) > \lambda$.

Let $\theta_0,\dots, \theta_m$ be a strictly increasing sequence such that $\ang{0} < \theta_i < \ang{90}$ for each $i$ and $\sum_{i=0}^m\theta_i=\ang{90}$. For instance, we can let $\theta_i=(\frac{180\cdot (i+1)}{m(m+2)})^\circ$ for $i=0,\dots,m$. 

Now we construct $\mathcal{A} = (\Sigma, Q, q^{[0]}_0, \ket{\psi_0}, (U_{a})_{a\in\Sigma}, (P_q)_{q\in Q}, \Delta)$, where:
\begin{enumerate}[$\bullet$]
\item $Q=\bigcup_{j=0}^m\{q^{[j]}_0,\dots, q^{[j]}_m\}$.
\item $\ket{\psi_0}=(1,0)^T$.
\item $\Sigma = \{\theta_0,\dots,\theta_m\}$, and for each $j$, $U_{\theta_j}$ is the unitary matrix corresponding to a rotation by angle $\theta_j$ in the 2-dimensional space.
\item $(q^{[j]}_s,\theta_t,q^{[j+1]}_t)$ is a transition in $\Delta$ if $(v_s,v_t)\in E$, for each $j=0,\dots,m-1$.
\item $P_q$ is the matrix
$\begin{pmatrix}
0 & 0\\
0 & 1
\end{pmatrix}$
if $q=q^{[m]}_m$; otherwise, $P_q$ is the null matrix for each $q\in Q$.
\end{enumerate}

By the construction of $\mathcal{A}$, if $\pi$ is a run from $(q^{[0]}_0,\ket{\psi_0})$ to $(q^{[m]}{m},\ket{\psi})$ for some $\psi$, then the length of $\pi$ must be $m$. Moreover, $\psi = (0,1)$ if and only if along $\pi$ each symbol $\theta_1,\dots,\theta_m$ appears exactly once; otherwise, $\psi=(\cos{\theta},\sin{\theta})$ for some $\theta < \ang{90} - \frac{\theta_0}{2}$ or $\theta > \ang{90} - \frac{\theta_0}{2}$.

The condition that each symbol $\theta_1,\dots,\theta_m$ appears exactly once along $\pi$ implies that $G$ has a path such that, excluding the starting node $v_0$, every node in ${v_1,\dots,v_m}$ is visited exactly once. By letting $\lambda = \frac{\sin{\theta_0}}{2}$, we have $\val_{\mathcal{A}}(w) > \lambda$ for some $w\in\Sigma^*$ if and only if $G$ contains a Hamiltonian path from $v_0$ to $v_m$. The cases for $\bowtie \in \{<,\geq,\leq, =\}$ can be derived similarly.
\qed
\end{proof}

\section{Missing Proofs for~\Cref{subsec:nip_encoding,subsec:relaxation}}

\CosLinearForm*

\begin{proof}
    By the Eigen-decomposition of $U$, we can write $U = V D V^{\dagger}$, where $V$ is a unitary matrix and $D$ is a diagonal matrix with diagonal entries being the eigenvalues of $U$. Further since $U$ is unitary, we can write $D = \mathrm{diag}(e^{i\theta_{1}},\ldots,e^{i\theta_{n}})$ for some $\theta_{1},\ldots,\theta_{n}\in\mathbb{R}$. Then we have
    \begin{align*}
        \norm{PU^{t}\ket{\psi}} &= \bra{\psi}(U^t)^{\dagger}P^{\dagger}PU^{t}\ket{\psi}\\
         &=\bra{\psi} V (D^{t})^{\dagger} V^{\dagger} P V D^{t} V^{\dagger}\ket{\psi}\\
         & \omit\hfill ($P^{\dagger}=P, PP = P$)\\
        \intertext{Let $M = V^{\dagger}PV$,$\ket{x}=V^{\dagger}\ket{\psi}$,}
         &=\bra{x}(D^{\dagger})^{t} M D^{t}\ket{x}\\
         &=\sum_{i,j\in [n]\times [n]} (\overline{x_{i}}e^{-{\bf i}\theta_{i}t}) M_{ij} (x_{j}e^{{\bf i}\theta_{j}t})\\
         &=\sum_{k\in [n]}\abs{x_{k}}^{2}M_{kk} + \sum_{1\leq i < j\leq n}\bigg(\overline{x_{i}}x_{j}M_{ij} e^{{\bf i}(\theta_{j}-\theta_{i})t} + \overline{x_{j}} x_{i} M_{ji} e^{{\bf i}(\theta_{i}-\theta_{j})t}\bigg)\\
         &=\sum_{k\in [n]}\abs{x_{k}}^{2}M_{kk} + \sum_{1\leq i < j\leq n}2\text{Re}\bigg(\overline{x_{i}}x_{j}M_{ij} e^{{\bf i}(\theta_{j}-\theta_{i})t}\bigg)\\
         &\omit\hfill ($M_{ji} = \overline{M_{ij}}$)\\
        \intertext{Let $C = \sum_{k\in [n]}\abs{x_{k}}^{2}M_{kk}$ and $\overline{x_{i}}M_{ij}x_{j} = \abs{\overline{x_{i}}M_{ij}x_{j}}e^{{\bf i}\phi_{ij}}$,}
         &= C + \sum_{1\leq i < j\leq n}2\abs{\overline{x_{i}}M_{ij}x_{j}}\cos((\theta_{j}-\theta_{i})t+\phi_{ij}).
    \end{align*}

    Hence we have $C=\sum_{k\in [n]}\abs{x_{k}}^{2}M_{kk}$, $a_{ij} = 2\abs{\overline{x_{i}}M_{ij}x_{j}}$, $(\theta_{i})_{i}$ are the phase of eigenvalue of $U$ and $(\phi_{ij})_{ij}$ are the phase of $\overline{x_{i}}M_{ij}x_{j}$ for each $1\leq i < j\leq n$. To effectively compute these parameters, it is obvious that the time complexity is dominiated by computing the eigen-decomposition of $U$, which can be done in time $O(n^{3})$ by the standard algorithm. \qed
\end{proof}

\LowerAndUpperBound*

\begin{proof}
    The proof can be found in most of the standard calculus textbooks, e.g. \cite{courant1965introduction}. Here we only give a sketch of the inequality part for self-containedness. Note that $\cos x, L_{m}(x)$ and $U_{m}(x)$ for any $m\in\mathbb{N}$ are even functions, therefore it suffices for us to consider only the case of $x\geq 0$. The inequality can be proven inductively on $m$. For the base case of $m=0$, starting with the well known inequality of $\sin x\leq x$ for $x\geq 0$, integrating both side from $0$ to $x$ gives us $1-\cos x\leq \frac{x^{2}}{2}$, further rearranging it gives $\cos x\geq 1-\frac{x^{2}}{2}$, which proven the lower bound part for $m=1$. Now starting from this lower bound and integrating both side twice, we have $1-\cos x \geq \frac{x^{2}}{2}-\frac{x^{4}}{4!}$, rearranging it gives $\cos x\leq 1-\frac{x^{2}}{2}+\frac{x^{4}}{4!}$. This completes the proof for the base case. The inductive step can be proven in the same way, starting from the upper bound of inductive hypothesis and integrating both side twice to get the lower bound of $m+1$ case, and integrating the lower bound twice to get the upper bound of $m+1$ case. \qed
\end{proof}

\soundness*

\begin{proof}
    It is a direct consequence of~\cref{prop:cos_linear_form,prop:cosine}. We show it for the case of $\bowtie$ being $>$, the case of $\bowtie$ being $<$ can be proven in the same way. Let 
    $$t_{\ast} = \argmax_{t\in\mathbb{N}}C + \sum_{1\leq i < j\leq n}a_{ij}L_{m}((\theta_{j}-\theta_{i})t+\phi_{ij}),$$
    then by the assumption we have
    $$C + \sum_{1\leq i < j\leq n}a_{ij}L_{m}((\theta_{j}-\theta_{i})t_{\ast}+\phi_{ij}) \bowtie \lambda^{2}.$$
    Further by~\cref{prop:cos_linear_form}, we know that $a_{ij}\geq 0$ for all $1\leq i < j\leq n$, hence replacing $L_{m}$ by $\cos$ preservse the ordering. Since we know that $\bowtie$ is either $>$ or $\geq$, therefore by~\cref{prop:cos_linear_form,prop:cosine} we have
    \begin{align*}
        \norm{P U^{t}\ket{\psi}} &= C + \sum_{1\leq i < j\leq n}a_{ij}\cos((\theta_{j}-\theta_{i})t_{\ast}+\phi_{ij})\\
        &\geq C + \sum_{1\leq i < j\leq n}a_{ij}L_{m}((\theta_{j}-\theta_{i})t_{\ast}+\phi_{ij}) \bowtie\lambda^{2}.
    \end{align*}\qed
\end{proof}

\completeness*

\begin{proof}
    We only show it for the case of $\bowtie\in\{>,\geq\}$, the case of $\bowtie\in\{<,\leq\}$ can be proven in the same way. By Lagrange remainder theorem, for any $x\in\mathbb{R}$ we have
    $$\cos x = L_{m}(x) + \frac{(\cos \alpha)x^{4m}}{(4m)!},$$
    for some $\alpha\in [0,x]$. Therefore we have 
    $$L_{m}(x)\geq \cos x - \frac{x^{4m}}{(4m)!}.$$
    Now assume for some $t\in\mathbb{N}$ and $\varepsilon>0$, we have
    $$\norm{P U^{t}\ket{\psi}} = C + \sum_{1\leq i < j\leq n}a_{ij}\cos((\theta_{j}-\theta_{i})t_{\ast}+\phi_{ij}) > \lambda^{2}+\varepsilon.$$
    Then for each $1\leq i < j\leq n$, let $\tau_{ij} = (\theta_{j}-\theta_{i})t+\phi_{ij}$. We define $M_{ij}$ to be the smallest natural number such that $\forall m\geq M_{ij}$, we have $a_{ij}\frac{\tau_{ij}^{4m}}{(4m)!}\leq \frac{\varepsilon}{n^{2}}.$ Such $M_{ij}$ must exist both $\lim_{m\to\infty}\frac{x^{4m}}{(4m)!} = 0$ for any $x\in\mathbb{R}$. Now let $M = \max_{1\leq i < j\leq n}M_{ij}$, then for any $m\geq M$, we have
    \begin{align*}
        C + \sum_{1\leq i < j\leq n}a_{ij}L_{m}((\theta_{j}-\theta_{i})t+\phi_{ij}) &\geq C + \sum_{1\leq i < j\leq n}\bigg(a_{ij}\cos(\tau_{ij}) - a_{ij}\frac{\tau^{4m}}{(4m)!}\bigg)\\
        &\geq  C + \sum_{1\leq i < j\leq n}a_{ij}\cos(\tau_{ij}) - \sum_{1\leq i < j\leq n}\frac{\varepsilon}{n^{2}}\\
        & = \norm{P U^{t}\ket{\psi}} - \frac{n(n-1)}{n^{2}}\frac{\varepsilon}{2} > \lambda^{2}+\frac{\varepsilon}{2} > \lambda^{2}.
    \end{align*}\qed
\end{proof}

\uniformcompleteness*

\begin{proof}
    Again we only show it for the case of $\bowtie\in\{>,\geq\}$, the case of $\bowtie\in\{<,\leq\}$ can be proven in the same way. Recall from the previous proof that by Lagrange remainder theorem, for any $x\in\mathbb{R}$ we have
    $$L_{m}(x)\geq \cos x - \frac{x^{4m}}{(4m)!},$$
    therefore for $x\in [-\pi,\pi]$, we further have
    $$L_{m}(x)\geq \cos x - \frac{\pi^{4m}}{(4m)!}.$$
    Now for any $\varepsilon>0$ and $1\leq i< j\leq n$, we let $M_{ij}$ be the smallest natural number such that $\forall m\geq M_{ij}$, we have $a_{ij}\frac{\pi^{4m}}{(4m)!}\leq \frac{\varepsilon}{n^{2}}.$ Then we let $M = \max_{1\leq i < j\leq n}M_{ij}$. 
    
    Now for any $t\in\mathbb{N}$, let $\tau_{ij} = (\theta_{j}-\theta_{i})t+\phi_{ij}$ and $K_{ij} = -\floor{\frac{\tau_{ij}}{2\pi}}$, where $\floor{x}$ is the largest integer that is less than or equal to $x$. It is then straightfoward to verify that $\tau_{ij}+2K_{ij}\pi\in [0,2\pi)$. In the case that $\tau_{ij}+2K_{ij}\pi\in (\pi,2\pi)$, we let $K_{ij}' = K_{ij}-1$, then we have $\tau_{ij}+2K_{ij}'\pi = \tau_{ij}+2K_{ij}\pi -2\pi\in (-\pi,0)$. Therefore, in any case we can always find some $K_{ij}'\in\mathbb{Z}$ such that $\tau_{ij}+2K_{ij}'\pi\in [-\pi,\pi]$. 9 Further the choice of $K_{ij}$ must be unique since the interval $[-\pi,\pi]$ has length $2\pi$. 
    
    Finally, for any $m\geq M$ and any $t\in\mathbb{N}$, we have
    \begin{align*}
        &C + \sum_{1\leq i < j\leq n}a_{ij}L_{m}((\theta_{j}-\theta_{i})t+\phi_{ij}+2K_{ij}\pi) \\
        &\geq C + \sum_{1\leq i < j\leq n}\bigg(a_{ij}\cos((\theta_{j}-\theta_{i})t+\phi_{ij}) - a_{ij}\frac{\pi^{4m}}{(4m)!}\bigg)\\
        &\geq  C + \sum_{1\leq i < j\leq n}a_{ij}\cos((\theta_{j}-\theta_{i})t+\phi_{ij}) - \sum_{1\leq i < j\leq n}\frac{\varepsilon}{n^{2}}\\
        & = \norm{P U^{t}\ket{\psi}} - \frac{n(n-1)}{n^{2}}\frac{\varepsilon}{2} > \lambda^{2}+\frac{\varepsilon}{2} > \lambda^{2}.
    \end{align*}\qed
\end{proof}

\section{Implementation and Experiments Details}\label{sec:app_experiments}

\begin{algorithm}[H]\footnotesize
\caption{NIP-based bounded witness seaching}\label{alg:nip}
\DontPrintSemicolon
\SetKwFunction{ComputeCoeff}{ComputeCoeff}
\SetKwProg{Fn}{Function}{:}{}
\KwIn{A one-loop \qfac ${\cal A} = (U,P,{\psi_{0}})$, a pair $(\bowtie,\lambda)\in\cmp\times [0,1]$, and a maximum degree of approximation $\textsc{maxdeg}$.}
\KwOut{An integer $t\in\mathbb{N}$}
$C, a_{ij},\theta_{i},\phi_{ij}\leftarrow$ \ComputeCoeff{$U$,$P$,$\ket{\psi_{0}}$}\\
\uIf{$\bowtie\in\{>,\geq\}$}{
    $\opt\leftarrow\max$\\$(F_{m})_{m}\leftarrow (L_{m})_{m}$
}\Else{$\opt\leftarrow\min$\\$(F_{m})_{m}\leftarrow (U_{m})_{m}$}
\For{$m$ from $1$ to \textsc{maxdeg}}{
    (t,Val)$\longleftarrow$ $\left(\begin{aligned}
        \opt_{t\in \mathbb{N}, (K_{ij})_{i,j}\in\mathbb{Z}}&\quad C + \sum_{1\leq i < j\leq n}a_{ij}F_{m}((\theta_{j}-\theta_{i})t+\phi_{ij}+2K_{ij}\pi)\\
        \text{subject to}&\quad -\pi\leq (\theta_{j}-\theta_{i})t+\phi_{ij}+2K_{ij}\pi\leq \pi\quad\forall i,j.
    \end{aligned}\right)$\\
    \lIf{$\text{Val}\bowtie \lambda$}{
        \Return $t$
    }
}
\Return $\text{None}$\\

\Fn{\ComputeCoeff{$U$,$P$,$\ket{\psi_{0}}$}}{
    $n\leftarrow \text{dim}(P)$\\
    $D,V\leftarrow$ EigenDecomposition($U$)\\
    $M\leftarrow V^{\dagger}PV$\\
    $x\leftarrow V^{\dagger}\ket{\psi_{0}}$\\
    $C\leftarrow\sum_{i=1}^{n}a_{ii}\theta_{i}$\\
    $a_{ij}\leftarrow 2\abs{\Bar{x_{i}}M_{ij}x_{j}}$\\
    $\theta_{i}\leftarrow \text{phase}(D_{ii})$\\
    $\phi_{ij}\leftarrow \text{phase}(\Bar{x_{i}}M_{ij}x_{j})$\\
    \Return $C, a_{ij},\theta_{i},\phi_{ij}$
}
\end{algorithm}

\paragraph{Models Details.}

The \texttt{Simple-Rotation} is a simple model that consists of a rotational matrix on a 2D plane. The angles of rotation are fixed to be $\frac{\pi}{4}$, $\frac{\pi}{6}$ and $\frac{\pi}{8}$ in the experiment. Our algorithm checks whether there is a $t$ such that the probability of reaching $\ket{1}$ from $\ket{0}$ is greater than $0.99$.

The \texttt{Grover} is a model for the well-known Grover's algorithm for $k$ qubits, the oracle, i.e. the solution to be searched, is chosen randomly for each trial of run. Our algorithm checks whether there is a $t$ such that the success probability of Grover's algorithm is greater than $0.99$. 

The \texttt{Quantum-Walk} is a modified version of quantum random walk model presented in~\cite{Dai:24:CAV} with $k$ qubits plus one ancilla qubit for coin tossing. The original formulation allows walking in both direction while we restrict it to be one direction. Our algorithm checks if there is a $t$ such that the probability of reaching a certain state $\ket{x}$ in the first $k$ qubits is greater than $0.25$, where $x$ is randomly generated for each trial of run.


\end{document}